\documentclass[11pt,draftcls,onecolumn]{IEEEtran}

\usepackage{color}
\usepackage{amsfonts}
\usepackage{amsmath}
\usepackage{amsthm}
\usepackage{amssymb}
\usepackage{alltt}
\usepackage{epsfig}
\usepackage{mathrsfs}
\usepackage{graphicx,ifthen}
\usepackage{subfigure}
\usepackage{esint}
\usepackage{enumerate}
\usepackage{cite}
\usepackage{float}

\newtheorem{proposition}{Proposition}
\newtheorem{lemma}{Lemma}
\newtheorem{definition}{Definition}
\newtheorem{corollary}{Corollary}

\newcommand{\mb}[1]{\mathbf{#1}}
\newcommand{\mc}[1]{\mathcal{#1}}
\newcommand{\mbb}[1]{\mathbb{#1}}

\newcommand\gfrac[2]{\genfrac{}{}{0pt}{}{#1}{#2}}
\newcommand{\wt}[1]{\widetilde{#1}}
\newcommand{\overbar}[1]{\mkern 2mu\overline{\mkern-2mu#1\mkern-2mu}\mkern 2mu}

\begin{document}

\title{Outage Capacity of Rayleigh Product Channels: a Free Probability Approach}

\date{} 

\author{Zhong~Zheng,~\IEEEmembership{Member,~IEEE,}
        Lu~Wei,~\IEEEmembership{Member,~IEEE,}
        Roland~Speicher,\\
        Ralf~M\"uller,~\IEEEmembership{Senior Member,~IEEE,}
        Jyri~H\"am\"al\"ainen,~\IEEEmembership{Member,~IEEE,}
        and~Jukka~Corander
\thanks{Z. Zheng and J. H\"am\"al\"ainen are with the Department
of Communications and Networking, Aalto University, Finland (e-mails: \{zhong.zheng,~jyri.hamalainen\}@aalto.fi). L. Wei and J. Corander are with the Department of Mathematics and Statistics, University of Helsinki, Finland (e-mails: \{lu.wei,~jukka.corander\}@helsinki.fi). R. Speicher is with the Faculty of Mathematics, Saarland University, Germany (e-mail: speicher@math.uni-sb.de). R. M\"uller is with the Institute for Digital Communications, Friedrich-Alexander-Universit\"at Erlangen-N\"urnberg, Germany (e-mail: mueller@lnt.de). This work was presented in part at 2014 International Zurich Seminar on Communications.
}
}

\markboth{Submitted to IEEE Transaction on Information Theory}%
{Zheng \MakeLowercase{\textit{et al.}}: Outage Capacity of Rayleigh Product Channels: a Free Probability Approach}

\maketitle

\begin{abstract}

The Rayleigh product channel model is useful in capturing the performance degradation due to rank deficiency of MIMO channels. In this paper, such a performance degradation is investigated via the channel outage probability assuming slowly varying channel with delay-constrained decoding. Using techniques of free probability theory, the asymptotic variance of channel capacity is derived when the dimensions of the channel matrices approach infinity. In this asymptotic regime, the channel capacity is rigorously proven to be Gaussian distributed. Using the obtained results, a fundamental tradeoff between multiplexing gain and diversity gain of Rayleigh product channels can be characterized by closed-form expression at any finite signal-to-noise ratio. Numerical results are provided to compare the relative outage performance between Rayleigh product channels and conventional Rayleigh MIMO channels.

\end{abstract}


\IEEEpeerreviewmaketitle

\section{Introduction}

Multi-Input Multi-Output~(MIMO) wireless communications have received considerable attention since it is seen as the most promising way to increase link level capacity. Extensive works have focused on the performance of MIMO channels assuming a rich scattering environment. Therein, the presumed models include full-rank independent Rayleigh or Rician MIMO channels. However, in certain environments the propagation may be subject to structural limits of fading channels caused by either insufficient scattering~\cite{GesbertTCOM2002, MullerTIT2002} or the so-called keyhole effect~\cite{ChizhikTWC2002}. These channels exhibit rank deficiency compared to the independent Rayleigh and Rician models. The MIMO model that captures these effects is referred to as the double-scattering channel~\cite{GesbertTCOM2002}. It is characterized by a matrix product involving three deterministic matrices (i.e., transmit, receiver, and scatterer correlation matrices), and two statistically independent complex Gaussian matrices. In a typical office environment, empirical measurements have been used to demonstrate the validity of the double-scattering channel model~\cite{Mul01}.

There exist a number of studies concerning the information-theoretic quantities of the double-scattering channels. \mbox{Shin \emph{et. al.}} derived an upper bound for the ergodic capacity~\cite[Th. III.3]{SL03} and an exact expression for a single keyhole channel~\cite[Th. III.4]{SL03}. The diversity-multiplexing tradeoff of the double-scattering channel was obtained in~\cite{Yang11}. The authors in~\cite{LL11} investigated the asymptotic Rayleigh-limit when one of the matrix dimensions approaches infinity. In this case, the double-scattering model reduces to an equivalent Rayleigh MIMO channel. Furthermore, if all matrix dimensions are large, the ergodic capacity has been obtained in \cite{MullerTIT2002} via numerical integration. Recently, an asymptotic expression for ergodic capacity of the double-scattering channels was derived in \cite{HoydisAsilomar2011}. Moreover, authors in \cite{WZTH13,AKW13,AJK13} derived the ergodic mutual information for finite dimensional channel matrices. However, all the above results are valid for ergodic channels, where each codeword has infinite length. For many practical communication systems such as WLANs~\cite{wlan}, the channels, albeit random, are slowly varying and the encoding/decoding process is subject to a delay constraint with moderate target packet error rates around $10^{-2}$--$10^{-1}$. The fading channel seen by each codeword are therefore non-ergodic. In this case, the ergodic capacity has no physical significance, whereas the outage capacity is a more relevant performance metric~\cite{OzarowTVT1994}. In literature, the outage capacity has been studied for conventional Rayleigh MIMO channels~\cite{MSS03,HochwaldTIT2004,KamathTIT2005,Tulino05,DM05,KazakopoulosOutCap2011,ChenTIT2012} as well as for Rician MIMO channels~\cite{SimonOutCap2006,KangOutCap2006,RieglerOutCap2010} via various random matrix techniques.

To the best of our knowledge, the outage capacity for the double-scattering channel has not been addressed in the most general form\footnote{Note that authors in~\cite{DM05} derived the outage probability of a similar MIMO model with random steering matrices at antenna arrays. However, these steering matrices are slowly varying compared to multi-path fading and considered as deterministic. Thus, the tools in~\cite{DM05} are not applicable here.}. It turns out to be a difficult random matrix theory problem. To gain insights into the outage behavior of the double-scattering channel, we consider a simplified channel model involving a product of two statistically independent complex Gaussian matrices, also known as the Rayleigh product channel. This channel model corresponds to the scenario where the antenna elements as well as the scattering objects are sufficiently separated and there is no spatial correlation at antenna arrays or between scatterers. To characterize the capacity fluctuations, we use the free probability theory for large dimensional random matrices~\cite{voiculescu86,SpeicherICM14,CMSS07}. By utilizing the second order Cauchy transform and \mbox{$R$-transform} machinery, we derive a compact expression for the asymptotic variance of the capacity of the Rayleigh product channel. We further show that the channel capacity distribution is asymptotically Gaussian by proving a Central Limit Theorem~(CLT) for the Linear Spectral Statistics (LSS) of the Rayleigh product ensemble. This result generalizes the CLT for correlated Wishart random matrices~\cite{Bai2004} and the CLT for Rayleigh product ensembles from polynomial LSS to generic analytic functions~\cite{Breuer2013}. The capacity distribution is then utilized to study the corresponding finite Signal-to-Noise-Ratio~(SNR) Diversity-Multiplexing Tradeoff~(DMT). The derived results in this paper are formally valid when the dimensions of the channel matrices grow to infinity. However, numerical simulations show that they serve as good approximations when the numbers of antennas and scatterers are comparable to practical systems.

The rest of the paper is organized as follows. In Section~\ref{secModel}, we give the channel model, the signal model as well as the MIMO capacity formulation. In Section~\ref{secMoment}, we study the second order eigenvalue fluctuations and the asymptotic capacity variance. The second order Cauchy transform of Rayleigh product ensembles is derived in Section~\ref{secCumu}. The CLT of the capacity of Rayleigh product channels is proved in Section~\ref{secCLT}. Based on this result, the approximations for outage probability and the finite-SNR DMT are calculated. In Section~\ref{secConc}, we conclude the main findings of the paper. Proofs of the technical results are provided in the Appendices.

\emph{Notations.} Throughout the paper, vectors are represented by lower-case bold-face letters, and matrices are represented by upper-case bold-face letters. The complex vector field with length $n$ is denoted as $\mbb{C}^n$. We use $\mc{CN}(0,\mb{A})$ to denote the zero-mean complex Gaussian vector with covariance matrix $\mb{A}$ and $\mb{I}_n$ is an $n\times n$ identity matrix. The superscript $(\cdot)^{\dag}$ denotes the matrix conjugate-transpose operation and $(\cdot)^{\mathrm{T}}$ is matrix transpose. We denote $\overbar{(\cdot)}$ as the complex conjugate operator. Denote $\mathrm{Tr}(\mb{A})$ as the trace of $n\times n$ matrix $\mb{A}$ and $\mathrm{tr}(\mb{A})$ as the normalized trace $\mathrm{Tr}(\mb{A})/n$. The notation $\mbb{E}[\cdot]$ denotes the expectation, and $\det(\cdot)$ denotes the matrix determinant.

\section{Rayleigh Product MIMO Channels}\label{secModel}

\subsection{Channel Model}

Consider a discrete-time, baseband MIMO system with $T$ transmit and $R$ receive antennas. The channel is assumed to follow the Rayleigh product fading with $S$ scattering objects, as shown in Fig.~\ref{figDoubleScatter}. The channels between the $s$-th scatterer and transmit antennas are denoted by vector $\boldsymbol\theta_s = [\theta_{s1},\ldots,\theta_{sT}]$, and the channels between receive antennas and the $s$-th scatterer are denoted by vector $\boldsymbol\psi_{s} = [\psi_{s1},\ldots,\psi_{sR}]$. The end-to-end equivalent channel matrix $\mb{H}$ is given by
\begin{equation}
\mb{H} = \frac{1}{\sqrt{R S}}\sum_{s=1}^S \boldsymbol\psi_s^{\dag}\boldsymbol\theta_s = \frac{1}{\sqrt{R S}}\mb{\Psi}^\dag\mb{\Theta},\label{eqH}
\end{equation}
where $\mb{\Theta} = \left[\boldsymbol\theta_1^{\dag},\ldots,\boldsymbol\theta_S^{\dag}\right]^{\dag}$ and $\mb{\Psi} = \left[\boldsymbol\psi_1^{\dag},\ldots,\boldsymbol\psi_S^{\dag}\right]^{\dag}$. We assume $\boldsymbol\theta_s\sim\mc{CN}(0,\mb{I}_T)$ and $\boldsymbol\psi_s\sim\mc{CN}(0,\mb{I}_R)$, where $\boldsymbol\theta_i$ and $\boldsymbol\psi_j$, $1\le i,j\le S$, are statistically independent. The channel $\mb{H}$ is thus modeled as a product of two independent complex Gaussian random matrices. In line with~\cite{MullerTIT2002,HoydisAsilomar2011,WZTH13,AKW13,AJK13}, the channel $\mb{H}$ is normalized by the constant $1/\sqrt{R S}$ so that the total energy of the channel is equal to an AWGN channel with an array gain $\mbb{E}[\mathrm{Tr}(\mb{H}\mb{H}^{\dag})] = \sum_{i,j} \mbb{E}[|H_{ij}|^2] = T$.

\begin{figure}[!t]
\centering
\includegraphics[width=4.5in]{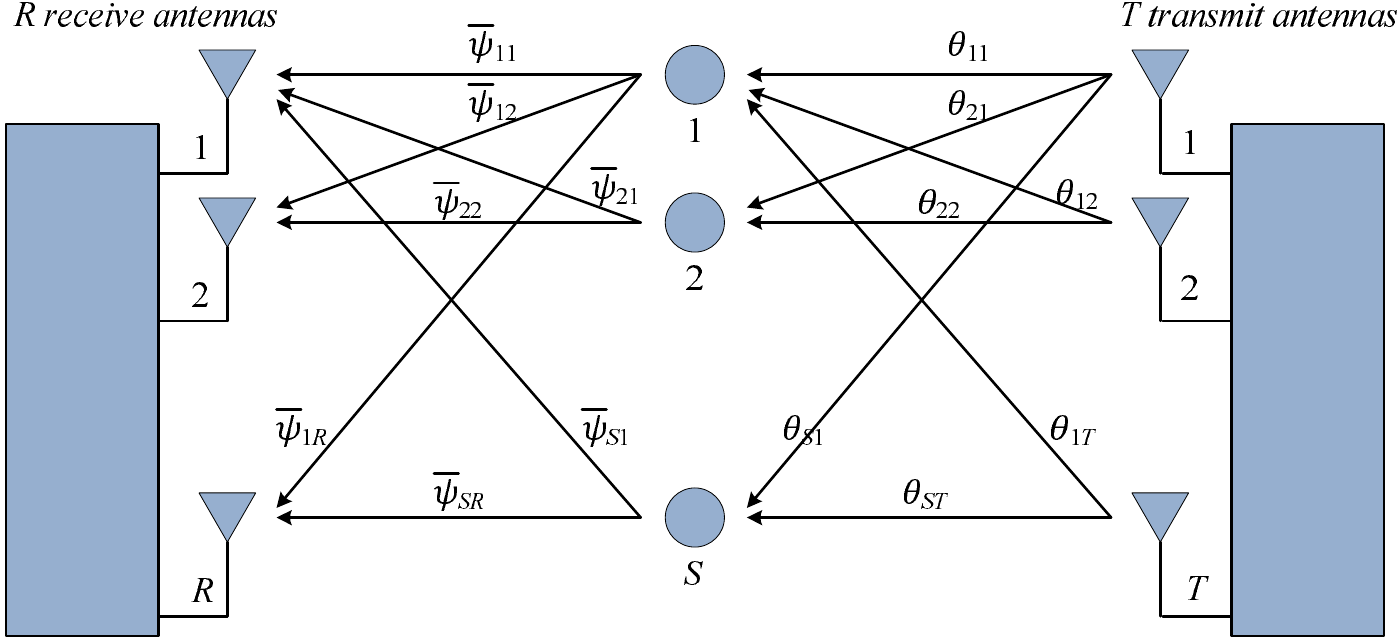}
\caption{MIMO communications over the Rayleigh product channel with $T$ transmit antennas, $R$ receive antennas, and $S$ scatterers.} \label{figDoubleScatter}
\end{figure}

The presence of independent Gaussian matrices $\mb{\Theta}$ and $\mb{\Psi}$ in (\ref{eqH}) requires two independent and richly-scattered environments, where the scattering happens between the $S$ scatterers/keyholes and transmit and receive arrays, respectively. This requires the existence of a large number of independently reflected and scattered paths around the antenna arrays~\cite{Tse2005}. The two environments are connected only via the $S$ scatterers/keyholes. By controlling the number $S$, the Rayleigh product channel (\ref{eqH}) embraces a general family of MIMO fading channel, spanning from the degenerate keyhole channel $S=1$ \cite{ChizhikTWC2002} to the full-rank Rayleigh MIMO channel $S\rightarrow\infty$ with fixed $R$ and $T$ \cite{LL11}. M\"uller and Hofstetter~\cite{Mul01} have shown that the number of significant scatterers is around ten in a typical office building with an $8\times 8$ antenna configuration. Measurement results in~\cite{AlmersTWC2006} indicate that the effective rank of a $6\times 6$ keyhole channel depends on the sizes of scatterer/keyhole at different transmission frequencies. In general, the number of separable scattering objects depends on the number of antenna elements since a larger array increases the spatial resolution~\cite{MullerTIT2002}. Note that the model (\ref{eqH}) also describes the MIMO relay channels when assuming noiseless relays~\cite{relay}.

\subsection{Signal Model and Channel Capacity}

The channel output vector $\mb{y}\in\mbb{C}^R$, at a given time instance, equals
\begin{equation}\label{eqY}
\mb{y} = \mb{H}\mb{x}+\mb{n},
\end{equation}
where $\mb{x}\in\mbb{C}^T$ is the transmit vector that follows the complex Gaussian distribution $\mb{x}\sim\mathcal{CN}(\mb{0},\mb{\Sigma})$ with $\mb{\Sigma}=\mbb{E}[\mb{x}\mb{x}^\dag]$. The additive noise $\mb{n}\in\mbb{C}^R$ is modeled as an \emph{i.i.d.} complex Gaussian vector $\mb{n}\sim\mc{CN}(\mb{0},\mb{I}_R)$. In this paper, we have adopted the following assumptions:
\begin{enumerate}[{A}1)]
\item The Channel State Information~(CSI) is perfectly known at the receiver but not at the transmitter.\label{a1}
\item The channel $\mb{H}$ is frequency flat and quasi-static. It remains constant for certain symbol durations and takes independently a new value for each coherence time.\label{a2}
\item Delay-constrained encoding/decoding. The encoded transmit message has a finite block length and spread in time over no more than a maximum allowable decoding delay. We assume the length of a coding block is equal to one independently faded interval.\label{a3}
\end{enumerate}

Under~A\ref{a1}, Telatar~\cite{Telatar1999} has shown that the channel capacity is achieved when the transmitted symbols are independent across antennas and the power is equally allocated, i.e. $\mb{\Sigma} = \gamma\mb{I}_T$ and $\gamma$ denotes the Signal to Noise Ratio~(SNR) per received antenna. The instantaneous capacity of the MIMO channel (\ref{eqY}) in nats/sec/Hz is given by
\begin{equation*}
\mc{I} =\log\det\left(\mb{I}_R+\gamma\mb{H}\mb{H}^{\dag}\right)=\sum_{i=1}^R \log(1+\gamma\lambda_i),
\end{equation*}
where $\lambda_i$, $i=1,\ldots,R$, refer to the eigenvalues of $\mb{Q} = \mb{H}\mb{H}^{\dag}$. For the Hermitian matrix $\mb{Q}$, we find it convenient to introduce the Empirical Spectral Distribution~(ESD) defined as
\begin{equation*}
\widetilde{F_\mb{Q}}(\lambda) = \frac{1}{R}\sum_{i=1}^R \mathbf{1}(\lambda_i\le \lambda),
\end{equation*}
where $\mathbf{1}(\cdot)$ denotes the indicator function. By letting $\varphi(x)=\log\left(1+\gamma x\right)$, the channel capacity $\mc{I}$ can be rewritten in terms of $\widetilde{F_{\mb{Q}}}(\lambda)$ as
\begin{equation}
\mc{I}=R\int \varphi(\lambda)\,\mathrm{d}\widetilde{F_\mb{Q}}(\lambda).\label{eqMI}
\end{equation}
As the channel matrix $\mb{H}$ is random, the instantaneous capacity~(\ref{eqMI}) is also a random variable. Without CSI at the transmitter, there is a non-zero probability, independent of the code length, that the channel capacity~(\ref{eqMI}) falls below any positive rate. Due to the assumptions A\ref{a2} and A\ref{a3}, the error probability corresponding to this rate cannot be decreased exponentially with the code length~\cite{OzarowTVT1994}. In this case, no reliable transmission is possible and the performance cannot be evaluated using the ergodic capacity. Instead, the fundamental performance limit of such a system is best explained with the capacity versus outage tradeoff, characterized by the Cumulative Distribution Function~(CDF) of $\mc{I}$. Given a fixed rate $r$, the outage probability is defined as the probability that capacity $\mc{I}$ is less than $r$, i.e.
\begin{equation}\label{eqPout}
P_{\mathrm{out}}(r) = \mathrm{Pr}\{\mc{I}\le r\} = F_{\mc{I}}(r),
\end{equation}
where $F_{\mc{I}}(\cdot)$ denotes the CDF of $\mc{I}$. When the CDF $F_{\mc{I}}(\cdot)$ is monotonically increasing, the outage capacity for a given probability $P_{\mathrm{out}}$ is obtained as
\begin{equation*}\label{eqIout}
\mc{I}_{\mathrm{out}} = F_{\mc{I}}^{-1}(P_{\mathrm{out}}).
\end{equation*}

The outage probability (\ref{eqPout}) is achievable~\cite{PrasadTIT2006} in the sense that for any $\epsilon > 0$, there exists a code of sufficiently large block length for which the packet error rate is upper-bounded by $P_{\mathrm{out}}(r)+\epsilon$. Thus, outage capacity provides useful insights on the operational performance of a delay-constrained coded system. Outage probability is also a meaningful metric to characterize the performance of some contemporary communication systems~\cite{LozanoMag2012}, where timely CSI is available at the transmitters. From this viewpoint, the complementary outage probability $1-P_{\mathrm{out}}(r)$ can be interpreted as the percentage of time that a transmission takes place at given rate $r$ under perfect link adaptation.

\section{Statistics of Channel Capacity}\label{secMoment}

In this section, we first review the convergence of empirical eigenvalue distribution of the Hermitian matrix $\mb{Q}=\mb{H}\mb{H}^\dag$ when matrix dimensions grow to infinity. The capacity per receive antenna is shown to converge to a deterministic value and expressed by a known result in~\cite{HoydisAsilomar2011}. Then, we study the global fluctuation of eigenvalues around the limiting distribution by deriving a closed-form expression for the the second order Cauchy transform of~$\mb{Q}$. This result is utilized to obtain the asymptotic variance of the channel capacity.

\subsection{First Order Cauchy Transform and Asymptotic Capacity}

For an $N\times N$ Hermitian random matrix $\mb{A}$, we assume that its ESD $\widetilde{F_{\mb{A}}}(\cdot)$ converges to a non-random limiting distribution $F_{\mb{A}}(\cdot)$ as $N\rightarrow\infty$. Such a convergence is alternatively established in~\cite{Bai1993} via the  convergence of resolvent $\wt{\mc{G}_{\mb{A}}}(z)$ to the first order Cauchy transform\footnote{In what follows, we refer to the first order Cauchy transform simply as Cauchy transform unless otherwise stated.} $\mc{G}_{\mb{A}}(z)$, defined as
\begin{align}
\wt{\mc{G}_{\mb{A}}}(z) &= \mathrm{tr}\left(\mb{I}_N z - \mb{A}\right)^{-1} = \int\frac{1}{z-t}\mathrm{d}\wt{F_{\mb{A}}}(t),\nonumber\\
\mc{G}_{\mb{A}}(z) &= \int_{\mc{S}_{\mb{A}}} \frac{1}{z-t}\mathrm{d}F_{\mb{A}}(t).\label{eqResolvent}
\end{align}
Here, $z\in\mbb{C}^{+}=\{z: \mathrm{Im}(z)>0\}$ and $\mc{S}_{\mb{A}}$ denotes the support of $F_{\mb{A}}(\cdot)$. Due to this limiting behavior of eigenvalues, the normalized linear spectral statistics, such as normalized capacity $\mc{I}/R$, converges to a non-random limit as $N\rightarrow\infty$ for a wide class of matrix ensembles~\cite{MullerTIT2002,MSS03,RieglerOutCap2010}. 

In the following, the limit $\lim\limits_{R\rightarrow\infty}$ denotes the asymptotic regime,
\begin{equation}\label{eqDim}
T,\ S,\ R\rightarrow\infty,\quad\mbox{with}\quad\rho = \frac{S}{R}\quad \mbox{and}\quad \zeta = \frac{T}{S}\quad \mbox{fixed}.
\end{equation}
In the asymptotic regime (\ref{eqDim}), Silverstein~\cite{Silverstein95} shows that the ESD $\wt{F_{\mb{Q}}}(\cdot)$ converges almost surely to a non-random CDF $F_{\mb{Q}}(\cdot)$  and its Cauchy transform $\mc{G}_{\mb{Q}}$ is the solution to
\begin{equation}
z = \frac{1}{\mc{G}_{\mb{Q}}} + \rho\int_{\mc{S}_{\mb{P}}}\frac{\lambda\mathrm{d}F_{\mb{P}}(\lambda)}{1-\lambda \mc{G}_{\mb{Q}}},\label{eqGQ}
\end{equation}
where $\mb{P} = \mb{\Theta}\mb{\Theta}^{\dag}/S$ and $F_{\mb{P}}(\cdot)$ is the well-known Mar\v{c}enko-Pastur distribution~\cite{MP67}. The integration range in (\ref{eqGQ}) is $\left[(1-\sqrt{\zeta})^2,\ (1+\sqrt{\zeta})^2\right]$. Using multiplicative free convolution, M\"uller~has shown in~\cite{MullerTIT2002} that $\wt{\mc{G}_{\mb{Q}}}(z)\rightarrow\mc{G}_{\mb{Q}}(z)$ in the asymptotic regime (\ref{eqDim}) and $\mc{G}_{\mb{Q}}$ satisfies the cubic equation
\begin{equation}
z^2 \mc{G}_\mb{Q}^3(z) + (\rho\zeta + \rho - 2) z \mc{G}_\mb{Q}^2(z)+\big((\rho\zeta-1)(\rho-1)-\rho z\big)\mc{G}_\mb{Q}(z) + \rho=0.\label{eqcubicGQ}
\end{equation}
If $|z|\rightarrow\infty$, $\mc{G}_{\mb{Q}}(\cdot)$ admits the formal power series expansion
\begin{equation}\label{eqGz}
\mc{G}_{\mb{Q}}(z) = \sum_{n=0}^{\infty} \alpha_n z^{-n-1},
\end{equation}
where $\alpha_0 = 1$ and $\alpha_n$ is the $n$-th free moment of $\mb{Q}$, defined as
\begin{equation*}
\alpha_n = \lim_{R\rightarrow\infty}\mbb{E}\left[\mathrm{tr}(\mb{Q}^n)\right].
\end{equation*}
A concept closely related to the Cauchy transform $\mc{G}_{\mb{Q}}(z)$ is the \mbox{$R$-transform} $\mc{R}(z)$, defined as a functional~\cite{voiculescu86}
\begin{equation}\label{eqGRfunc}
\mc{G}_{\mb{Q}}\left(\mc{R}(z)+\frac{1}{z}\right) = z.
\end{equation}
If $|z|\rightarrow\infty$, $\mc{R}(z)$ has the formal power series representation
\begin{equation}\label{eqRz}
\mc{R}(z) = \sum_{n=1}^{\infty} \kappa_n z^{n-1},
\end{equation}
where $\kappa_n$ is the $n$-th free cumulant of $\mb{Q}$. As will be shown in Section~\ref{secCumu}, the free cumulant sequence $\{\kappa_n\}_{n\ge 1}$ and its generating function~$\mc{R}(z)$ serve as the key analytical tools in the proof of Proposition~\ref{propCauchy2}. Note that when the matrix dimensions are finite, the computation of $\alpha_n$ involves non-trivial summations over all partitions of integer $n$~\cite[Eq. (27)]{LuOSTBC2014}. This complicated expression makes it challenging to obtain an explicit expression for the free cumulant $\kappa_n$. As the matrix dimensions approach infinity, the calculation of $\kappa_n$ is much simplified, involving only the so-called non-crossing permutations over integers, see Lemma~\ref{lemmaCumu} and Appendix~\ref{appxNC} for a detailed discussion. 


Using random matrix theory techniques, authors in~\cite{HoydisAsilomar2011} prove that the capacity per receive antenna of the Rayleigh product channel converges to an asymptotic limit such that $\lim\limits_{R\rightarrow\infty} \left(\mc{I}-\mu_{\mc{I}}\right)/R=0$. The asymptotic capacity $\mu_{\mc{I}}/R$ is given by an explicit closed-form expression, which is summarized in the following proposition.

\begin{proposition}(Asymptotic capacity \cite{HoydisAsilomar2011})\label{propMu}
When $R = T$, the asymptotic capacity per receive antenna~$\mu_\mc{I}/R$ in the regime (\ref{eqDim}) is given by
\begin{equation*}
\frac{\mu_{\mc{I}}}{R} = \log\left(\frac{1}{g}+\frac{\gamma}{\rho}(g+\rho-1)\right)-\rho\log\left(1+\frac{g-1}{\rho}\right)-2(1-g),
\end{equation*}
where $g$ is the unique solution to
\begin{equation*}
g^3-(1-\rho)g^2+\frac{\rho}{\gamma}(g-1) = 0
\end{equation*}
such that $(1-g)/(g(g+\rho-1))\ge 0$.
\end{proposition}
Although the asymptotic capacity $\mu_{\mc{I}}$ grows to infinity in the asymptotic regime (\ref{eqDim}), it serves as a tight approximation to the mean capacity $\mbb{E}[\mc{I}]$ with finite matrix dimensions as shown in~\cite{HoydisAsilomar2011}. In the following, we will also use $\mu_{\mc{I}}$ as the approximated $\mbb{E}[\mc{I}]$ whenever it is clear from the context.

\subsection{Second Order Cauchy Transform and Asymptotic Variance}\label{secVar}

As $\wt{F_{\mb{Q}}}(\cdot)\rightarrow F_{\mb{Q}}(\cdot)$ in the asymptotic regime~(\ref{eqDim}), the asymptotic capacity per-received antenna $\mu_{\mc{I}}/R$ can be formulated by replacing $\wt{F_{\mb{Q}}}(\lambda)$ in (\ref{eqMI}) with $F_{\mb{Q}}(\lambda)$. Using an integral identity\footnote{The definition of Stieltjes transform in \cite{Bai2004} is different from the Cauchy transform by a minus sign.}~\cite[Eq. (1.14)]{Bai2004}, we utilize an amenable form of asymptotic capacity, which is useful in the following discussion, namely
\begin{equation}
\frac{\mu_{\mc{I}}}{R} = \int_{\mc{S}_{\mb{Q}}} \varphi(\lambda)\mathrm{d}F_{\mb{Q}}(\lambda) = \frac{1}{2\pi \imath}\oint_{\mc{C}}\varphi(z)\mc{G}_{\mb{Q}}(z)\mathrm{d}z.\label{eqContourInt}
\end{equation}
The complex integral on the right hand side of (\ref{eqContourInt}) is over any positively oriented closed contour $\mc{C}$ enclosing the support $\mc{S}_{\mb{Q}}$ and on which $\varphi(\cdot)$ is analytic. For the instantaneous channel capacity (\ref{eqMI}), there exists a similar integral expression as in (\ref{eqContourInt}) with Cauchy transform $\mc{G}_{\mb{Q}}(z)$ replaced with the resolvent $\wt{\mc{G}_{\mb{Q}}}(z)$. To see this, let the contour $\mc{C}$ be selected according to (\ref{eqContourInt}) and apply Cauchy's integral formula on $\varphi(\lambda)$, it follows that the instantaneous capacity~(\ref{eqMI}) becomes
\begin{align*}
\frac{\mc{I}}{R} = \int \varphi(\lambda)\,\mathrm{d}\widetilde{F_\mb{Q}}(\lambda)
= \frac{1}{2\pi \imath}\int \oint_{\mc{C}}\frac{\varphi(x)}{x-\lambda}\mathrm{d}x\,\mathrm{d}\widetilde{F_\mb{Q}}(\lambda).
\end{align*}
Exchange the integrations and recall the definition of resolvent (\ref{eqResolvent}), we obtain
\begin{equation}
\frac{\mc{I}}{R} = \frac{1}{2\pi \imath}\oint_{\mc{C}}\varphi(x)\wt{\mc{G}_{\mb{Q}}}(x)\mathrm{d}x.\label{eqMIb}
\end{equation}
Let us now consider the variance of capacity $\mc{I}$, defined as $\sigma_{\mc{I}}^2 = \mbb{E}[(\mc{I}-\mbb{E}[\mc{I}])^2]$. Replacing $\mc{I}$ with (\ref{eqMIb}) and $\mbb{E}[\mc{I}]$ with (\ref{eqContourInt}), the variance $\sigma_{\mc{I}}^2$ can be rewritten as
\begin{equation*}
\sigma_{\mc{I}}^2 = -\frac{1}{4\pi^2}\mbb{E}\left[\oint_{\mc{C}_x}\varphi(x)G_R(x)\mathrm{d}x\,\oint_{\mc{C}_y}\varphi(y)G_R(y)\mathrm{d}y\right],
\end{equation*}
where $G_R(x) = R\left(\wt{\mc{G}_{\mb{Q}}}(x)-\mc{G}_{\mb{Q}}(x)\right)$, the contours $\mc{C}_x$ and $\mc{C}_y$ are non-overlapping and are taken in the same way as in (\ref{eqContourInt}). After interchanging the expectation and integrations, we have
\begin{equation}
\sigma_{\mc{I}}^2 = -\frac{1}{4\pi^2} \oiint_{\mc{C}_x,\mc{C}_y}\varphi(x)\varphi(y)\mathrm{Cov}\left(G_R(x), G_R(y)\right)\mathrm{d}x\mathrm{d}y,\label{eqVar}
\end{equation}
where $\mathrm{Cov}\left(G_R(x), G_R(y)\right) = \mbb{E}\left[G_R(x)\cdot G_R(y)\right]$ is the covariance function of matrix resolvent scaled by the matrix dimension $R$. In the context of free probability theory, this covariance function is known as the second order Cauchy transform~\cite{CMSS07} and is denoted as~$\mathrm{Cov}\left(G_R(x), G_R(y)\right)\triangleq\mc{G}_\mb{Q}(x,y)$.

The rest of this section is devoted to derive the second order Cauchy transform of $\mb{Q}$ and the asymptotic variance of capacity $\sigma_{\mc{I}}^2$ by using a recent result from free probability theory. Namely, by the framework of the second order freeness~\cite{CMSS07,AMpreprint}, the second order Cauchy transform $\mc{G}_{\mb{Q}}(x,y)$ exists if $\mb{Q}$ has a second order limiting distribution according to the following definition: 
\begin{definition}\label{defSecDist}
Let $\mb{A}_N$ be an $N\times N$ random matrix. We say that it has a second order limiting distribution if for all $m,n\ge 1$ the moments $\{\alpha_n\}_{n\ge 1}$ and the limits
\begin{align*}
\alpha_{m,n} = \lim_{N\rightarrow\infty} k_2\left(\mathrm{Tr}(\mb{A}_N^m),\mathrm{Tr}(\mb{A}_N^n)\right)
\end{align*}
exist and if for all $r\ge 3$ and all $n(1),\ldots,n(r)\ge 1$,
\begin{equation*}
\lim_{N\rightarrow\infty} k_r\left(\mathrm{Tr}\left(\mb{A}_N^{n(1)}\right),\ldots,\mathrm{Tr}\left(\mb{A}_N^{n(r)}\right)\right) = 0,
\end{equation*}
where $k_r$ denotes the $r$-th classic cumulant.
\end{definition}
As $\mb{\Theta}\mb{\Theta}^{\dag}/S$ and $\mb{\Psi}\mb{\Psi}^{\dag}/R$ are independent complex Wishart matrices, they are unitarily invariant and their second order limiting distributions exist~\cite[Th. 3.5]{MS06}. It follows from \cite[Eq. (29)]{CMSS07} and the cyclic invariant property of matrix trace that $\mb{Q}$ also has the second order limiting distribution. The second order Cauchy transform $\mc{G}_{\mb{Q}}(x,y)$ is given by the functional~\cite[Eq. (53)]{CMSS07}
\begin{align}\label{eqSecCauchy}
\mc{G}_\mb{Q}(x,y)=\mc{G}_\mb{Q}'(x) & \mc{G}_\mb{Q}'(y)\mc{R}(\mc{G}_\mb{Q}(x),\mc{G}_\mb{Q}(y))+\frac{\partial^2}{\partial x \partial y}\log\frac{\mc{G}_\mb{Q}(x)-\mc{G}_\mb{Q}(y)}{x-y},
\end{align}
where $\mc{R}(x,y)$ denotes the second order \mbox{$R$-transform} of $\mb{Q}$. Similar to the first order case, if $|x|\rightarrow\infty$ and $|y|\rightarrow\infty$, $\mc{G}_{\mb{Q}}(x,y)$ and $\mc{R}(x,y)$ have formal power series representations
\begin{align}
\mc{G}_{\mb{Q}}(x,y) = \sum_{m,n\ge 1}\alpha_{m,n}x^{-m-1}y^{-n-1},\quad\mathrm{and}\quad\mc{R}(x,y) = \sum_{m,n\ge 1}\kappa_{m,n}x^{m-1}y^{n-1}.\label{eqRxy}
\end{align}

In literature, the covariance function $\mc{G}_{\mb{B}}(x,y)$ for the Wishart type $N\times N$ random matrix $\mb{B} = (1/N) \mb{X}^{\dag}\mb{T}\mb{X}$ has been studied in~\cite{Bai2004}, where $\mb{T}$ is a non-random Hermitian matrix and $\mb{X}$ is a Gaussian like\footnote{Each entry of the Gaussian like matrix has the same second and fourth moments as a Gaussian random variable.} random matrix with \emph{i.i.d.} entries. Therein, the correlation function of $\mb{B}$ has the form
\begin{equation}\label{eqGB}
\mc{G}_{\mb{B}}(x,y) = \frac{\mc{G}_{\mb{B}}'(x)\mc{G}_{\mb{B}}'(y)}{\left(\mc{G}_{\mb{B}}(x)-\mc{G}_{\mb{B}}(y)\right)^2}-\frac{1}{(x-y)^2},
\end{equation}
and it is subsequently used to derive an asymptotic variance of Rayleigh MIMO capacity in~\cite{KamathTIT2005,Tulino05,DM05}. Note that the second term of the right hand side of (\ref{eqSecCauchy}) is exactly the same as (\ref{eqGB}) by replacing $\mb{B}$ with $\mb{Q} = (1/R)\mb{\Psi}^{\dag}\mb{P}\mb{\Psi}$ and assuming $\mb{P}=\mb{\Theta}\mb{\Theta}^{\dag}/S$ non-random. Therefore, the fluctuation of capacity $\sigma_{\mc{I}}^2$ of Rayleigh product channels has a distinct functional structure from the Rayleigh MIMO channels, see (\ref{eqVar}) and (\ref{eqSecCauchy}). The increased fluctuation is due to a non-zero \mbox{$R$-transform} $\mc{R}(x,y)$. Closed-form expressions of $\mc{G}_{\mb{Q}}(x,y)$ and $\mc{R}(x,y)$ are summarized in the following proposition:
\begin{proposition}\label{propCauchy2}
The second order Cauchy transform of $\mb{Q}$ is given by (\ref{eqSecCauchy}) with
\begin{align}
\mc{R}(x,y) =\frac{\mc{G}_\mb{P}'(1/x) \mc{G}_\mb{P}'(1/y)}{x^2 y^2 (\mc{G}_\mb{P}(1/x)-\mc{G}_\mb{P}(1/y))^2}-\frac{1}{(x-y)^2},\label{eqR}
\end{align}
where $\mc{G}_{\mb{P}}(z)$ is the Cauchy transform of a Mar\v{c}enko-Pastur distribution with the parameter $\zeta$ as in (\ref{eqDim})
\begin{equation}
\mc{G}_\mb{P}(z) = \frac{1}{2}+\frac{1-\zeta}{2z} - \sqrt{\frac{1}{4}-\frac{1+\zeta}{2z}+\frac{(1-\zeta)^2}{4z^2}}.\label{eqGp}
\end{equation}
\end{proposition}
\begin{proof}
The proof of Proposition~\ref{propCauchy2} depends on the combinatorial structure of cumulants $\{\kappa_n\}_{n\ge 1}$ and $\{\kappa_{m,n}\}_{m,n\ge 1}$ and is given in detail in Section~\ref{secCumu}.
\end{proof}

Substitute (\ref{eqSecCauchy}) into (\ref{eqVar}) and denote $G(x) = \mc{G}_{\mb{P}}(1/\mc{G}_{\mb{Q}}(x))$, then the asymptotic variance $\sigma_{\mc{I}}^2$ is rewritten as
\begin{equation}
\sigma_{\mc{I}}^2 = -\frac{1}{4\pi^2}\oiint_{\mc{C}_x,\mc{C}_y} \frac{\varphi(x)\varphi(y)}{(G(x)-G(y))^2}\mathrm{d}G(x)\mathrm{d}G(y).\label{eqVarc}
\end{equation}
In the general setting, it is difficult to further simplify the double integral (\ref{eqVarc}). However, when the transmitter and receiver have equal number of antennas, i.e. $\zeta = 1/\rho$, a compact expression for $\sigma_{\mc{I}}^2$ can be obtained. The results are summarized in the following proposition.

\begin{proposition}\label{propVar}
When $R = T$, the asymptotic variance $\sigma_{\mc{I}}^2$ is given by
\begin{equation}
\sigma_{\mathcal{I}}^2 = \log\frac{\gamma(\omega_r-1)^2}{\gamma-\omega_r^2(2\omega_r-2)},\label{eqVars}
\end{equation}
where $\omega_r\le 0$ is the solution of the cubic equation
\begin{equation}\label{eqCubic}
t^3 -2 t^2 + (1-\gamma+\gamma\zeta) t + \gamma = 0.
\end{equation}
\end{proposition}
\begin{proof}
The proof of Proposition~\ref{propVar} is in Appendix~\ref{appxAsyVar}.
\end{proof}
Using Cardano's formula to solve the cubic equation~(\ref{eqCubic}), the explicit expressions for the roots are
\begin{align}
t_1 &= \frac{2}{3}-\frac{3\gamma\zeta-3\gamma-1}{3u(\gamma,\zeta)}+\frac{u(\gamma,\zeta)}{3}\label{eqt1},\\
t_2 &= \frac{2}{3}+e^{\frac{\imath\pi}{3}}\frac{3\gamma\zeta-3\gamma-1}{3u(\gamma,\zeta)}-e^{-\frac{\imath\pi}{3}}\frac{u(\gamma,\zeta)}{3}\label{eqt2},\\
t_3 &= \frac{2}{3}+e^{-\frac{\imath\pi}{3}}\frac{3\gamma\zeta-3\gamma-1}{3u(\gamma,\zeta)}-e^{\frac{\imath\pi}{3}}\frac{u(\gamma,\zeta)}{3}\label{eqt3},
\end{align}
where
\begin{equation}
u(\gamma,\zeta) = \left(\sqrt{(3\gamma\zeta-3\gamma-1)^3+\left(1+\frac{9\gamma}{2}+9\gamma\zeta\right)^2}-1-\frac{9\gamma}{2}-9\gamma\zeta\right)^{1/3}\label{equ}.
\end{equation}
For general values of $\gamma$ and $\zeta$, it is not straightforward to gain insights based on the variance expressions (\ref{eqVars}), (\ref{eqt1})-(\ref{equ}). However, in the high SNR regime with $\gamma\gg 1$, the asymptotic variance $\sigma_{\mc{I}}^2$ is characterized by explicit expressions and the behavior of capacity fluctuation can be understood.
\begin{corollary}\label{corollary1}
In the high SNR regime $\gamma\gg 1$, the asymptotic variance $\sigma_{\mc{I}}^2$ is given by (\ref{eqVars}) with $\omega_r$ approximated by
\begin{equation*}
\omega_r \approx \left\{\begin{array}{ll}\frac{1}{1-\zeta} & \zeta>1 \\
\frac{2}{3}-\sqrt{(1-\zeta)\gamma} & 0<\zeta<1.
\end{array}
\right.
\end{equation*}
For a fixed SNR, the variance of channel capacity is highest when the number of scattering objects equals to the number of antennas.
\end{corollary}
\begin{proof}
The proof of Corollary~\ref{corollary1} is in Appendix~\ref{appxCorollary1}.
\end{proof}


The asymptotic variance $\sigma_{\mc{I}}^2$ is derived with the assumption that the dimensions of matrices are large. However, $\sigma_\mc{I}^2$ serves as a good approximation for the variance of capacity even when the matrix dimensions are comparable to realistic MIMO systems. In Fig.~\ref{figVar}, we plot the variance of channel capacity as a function of the number of scattering objects~$S$ for $4\times 4$ and $8\times 8$ MIMO systems. The asymptotic variance is calculated by Proposition~\ref{propVar} at SNRs $\gamma$ ranging from $-20$~dB to~$20$~dB with a step size of $5$~dB. The analytical calculations are compared with Monte Carlo simulations, where each curve is generated by $10^6$ independent channel realizations. We also plot the asymptotic variance of a conventional Rayleigh MIMO channel using~\cite[Eq. (13)]{KamathTIT2005}. Fig.~\ref{figVar} shows that the asymptotic variance achieves a good agreement with the simulations for a wide range of SNRs and numbers of scatterings, especially in the low SNR regime. It is only when $\gamma>10$~dB that there are observable gaps between analytical and simulation curves. The asymptotic variance for a $8\times 8$ MIMO system remains a better approximation than that of a $4\times 4$ system, as expected. In the high SNR regime, see Fig.~\ref{figVar}\ (a), there exists a peak value for the variance of capacity when $S>1$. As the SNR $\gamma$ increases, the peak of the variance occurs at a fixed value $S = R = T$ ($\zeta = \rho = 1$), which is in line with our prediction in Corollary~\ref{corollary1}. This is analogous to the observations in~\cite{HochwaldTIT2004} that the capacity variance of the conventional Rayleigh MIMO channel is largest when $R=T$. On the other hand, the variance is monotonically decreasing in the low SNR regime, see Fig.~\ref{figVar}\ (b). As the number of scatterers becomes large, we also observe that the capacity variance of the Rayleigh product channel approaches a limit. This limit is set by the variance of conventional Rayleigh MIMO channel with the same number of antennas. This agrees with the results in~\cite{LL11}, where the multi-keyhole channel converges to a Rayleigh MIMO channel when the number of scatterers is large.

\begin{figure}[t!]
\centerline{\subfigure[]{\includegraphics[width=4in]{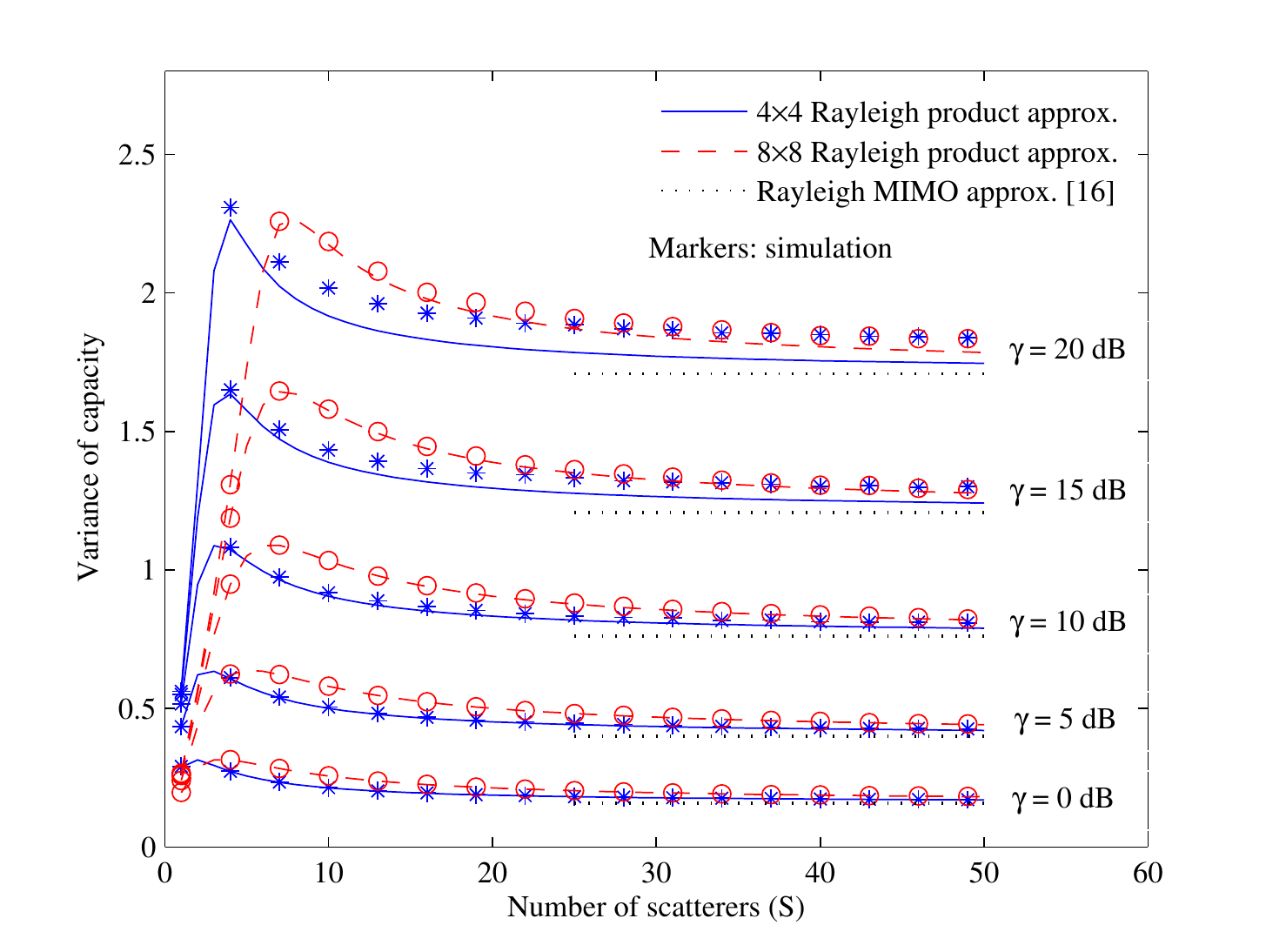}}\label{figVarH}}
\centerline{\subfigure[]{\includegraphics[width=4in]{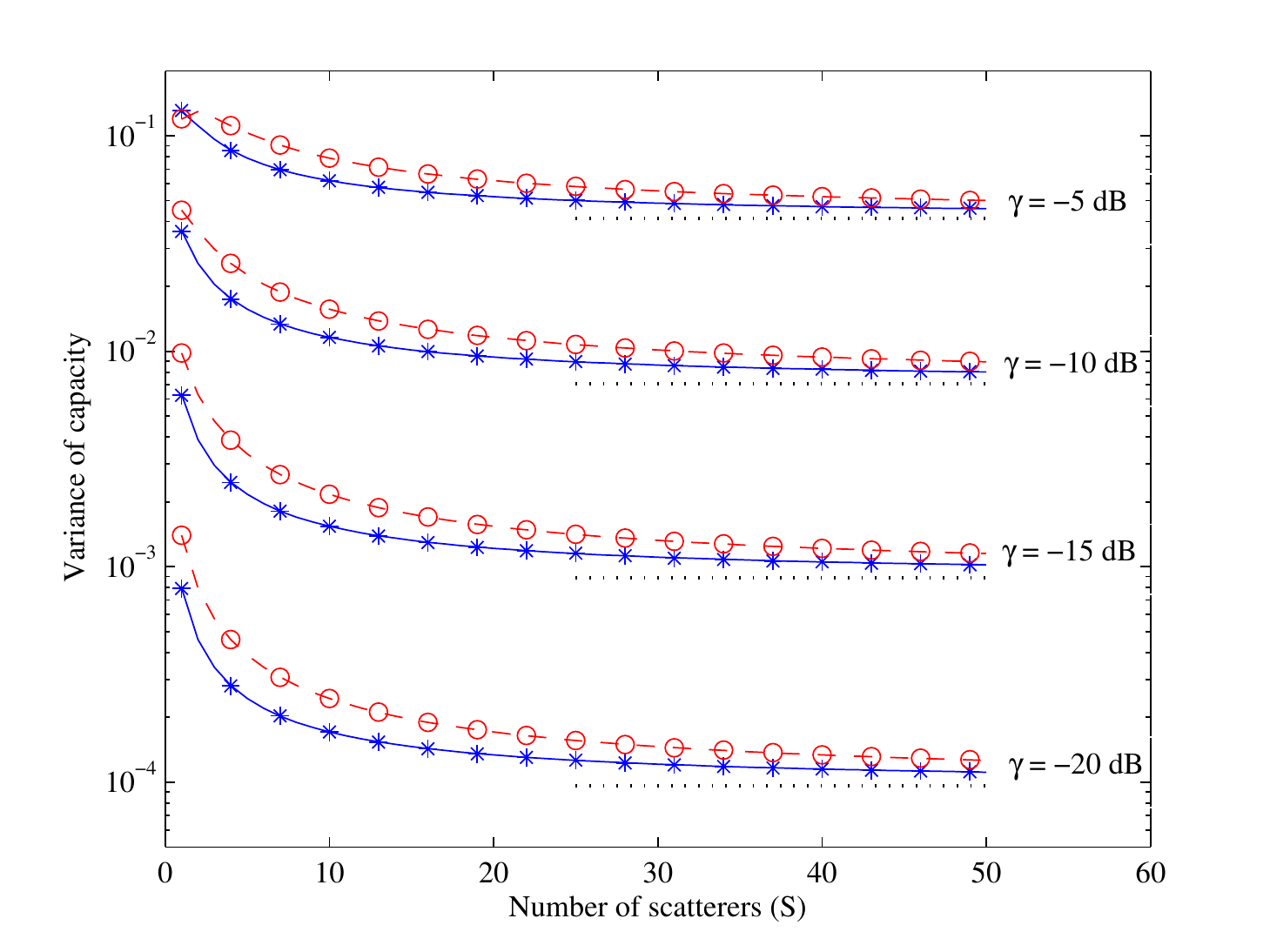}}\label{figVarL}}
\caption{Variance of the Rayleigh product channel capacity. Solid line: asymptotic variance of $4\times 4$ Rayleigh product channel~(\ref{eqVars}); dashed line: asymptotic variance of $8\times 8$ Rayleigh product channel~(\ref{eqVars}); dotted line: asymptotic variance of Rayleigh MIMO channel \cite{KamathTIT2005}; markers: simulation. (a) 0 dB$\le \gamma\le 20$ dB (b) $-20$ dB$\le\gamma\le -5$ dB.}
\label{figVar}
\end{figure}

\section{Second Order Cumulants and Cauchy Transform}\label{secCumu}

This section is devoted to the proof of Proposition~\ref{propCauchy2}, which relies on the knowledge of free cumulants of $\mb{Q}$. Let us recall from (\ref{eqRz}) and (\ref{eqRxy}) that $\mc{R}(z)$ and $\mc{R}(x,y)$ are generating functions for the cumulant sequences $\{\kappa_n\}_{n\ge 1}$ and $\{\kappa_{m,n}\}_{m,n\ge 1}$, respectively. We will first deduce the combinatorial descriptions of $\kappa_n$ and $\kappa_{m,n}$. These results reveal that the cumulant sequences of $\mb{Q}$ have the same combinatorial structures as the moment sequences of a Mar\v{c}enko-Pastur distribution. Namely, $\mc{R}(x,y)$ can be obtained based on known results. The notations and terminologies used in the formulation of Lemma~\ref{lemmaCumu} and in the proof of Proposition~\ref{propCauchy2} are given in Appendix~\ref{appxNC}.
\begin{lemma}\label{lemmaCumu}
For integers $m,n\ge 1$, the first order free cumulant $\kappa_n$ of matrix $\mb{Q}$ is given by
\begin{equation}\label{eqCumu1}
\kappa_n = \rho\sum_{\tau\in \mc{S}_{\operatorname{d-nc}}(n)} \zeta^{\#(\tau)},
\end{equation}
and the second order free cumulant $\kappa_{m,n}$ is given by
\begin{equation}\label{eqCumu2}
\kappa_{m,n} = \sum_{\tau\in \mc{S}_{\operatorname{a-nc}}(m,n)} \zeta^{\#(\tau)},
\end{equation}
where $\mc{S}_{\operatorname{d-nc}}(n)$ and $\mc{S}_{\operatorname{a-nc}}(m,n)$ denote the set of non-crossing permutation in disc and annular sense.
\end{lemma}
\begin{proof}
The proof of Lemma \ref{lemmaCumu} is in Appendix~\ref{appxLemmaCumu}.
\end{proof}

By comparing (\ref{eqCumu1}) with the $n$-th moment $\beta_n$ of $\mb{P}$~\cite[Eq. (7.3)]{MN04}, we have $\kappa_n = \rho\,\beta_n$ with $n\ge 0$. The cumulant $\kappa_n$ can be viewed as the scaled version of the moment $\beta_n$, where the normalized trace has the normalization factor $R$ instead of the actual matrix dimension~$S$.
Similarly, by comparing (\ref{eqCumu2}) with the second order moment $\beta_{m,n}$ of $\mb{P}$~\cite[Eq. (7.5)]{MN04}, we have the second order cumulant $\kappa_{m,n} = \beta_{m,n}$, where no normalization is needed as in Definition~\ref{defSecDist}. Note that the normalization for the first order moments can be arbitrarily chosen without affecting the underlying combinatorial structures provided that the normalization factor grows at the same rate as the matrix dimensions. The functional relations between moments and cumulants as well as their generating functions $\mc{G}_{\mb{Q}}(\cdot)$,  $\mc{R}(\cdot)$ are therefore preserved. For notational simplicity, it is convenient to work with the properly scaled moment sequences~$\{\beta_n\}$ and~$\{\beta_{m,n}\}$. 

As $\beta_n$ is the $n$-th moment of a Mar\v{c}enko-Pastur distribution with parameter $\zeta$, by the moment-cumulant relation~\cite{NS06}, the corresponding cumulant $c_n$ equals $\zeta$. In addition, it follows from the second order moment-cumulant relation~\cite{CMSS07} that the moment $\beta_{m,n}$ can be expressed in terms of $c_n$ as well as the corresponding second order cumulant $c_{m,n}$
\begin{equation}\label{eqSecMoCumu}
\beta_{m,n} = \sum_{\pi\in\mc{S}_{\operatorname{a-nc}}(m,n)} \prod_{i=1}^r c_{n_i} + \sum_{\substack{\pi_1\in\mc{S}_{\operatorname{d-nc}}(m)\\ \pi_2\in\mc{S}_{\operatorname{d-nc}}(n) \\ |\mc{V}|=|\pi_1\times\pi_2|+1}} c_{m_k,n_l}\prod_{\gfrac{i=1}{i\neq k}}^r c_{m_i}\,\prod_{\gfrac{j=1}{j\neq l}}^t c_{n_j}.
\end{equation}
On the right hand side of (\ref{eqSecMoCumu}), $\pi\in\mc{S}_{\operatorname{a-nc}}(m,n)$ contains $r\ge 1$ orbits and the $i$-th orbit contains $n_i$ elements with $n_1+ \cdots + n_r = m+n$. In the second summation, $\pi_1\in \mc{S}_{\operatorname{d-nc}}(m)$ and $\pi_2\in \mc{S}_{\operatorname{d-nc}}(n)$ contain $r\ge 1$ and $t\ge 1$ orbits, respectively.  The $i$-th orbit of $\pi_1$ contains $m_i$ elements with $m_1+ \cdots + m_r = m$, and $j$-th orbit of $\pi_2$ contains $n_j$ elements with $n_1+\cdots + n_t=n$. The partition $\mc{V}$ is composed of elements from the $k$-th orbit of $\pi_1$ and the $l$-th orbit of $\pi_2$ and corresponds to the second order cumulant $c_{m_k,n_l}$. Inserting $c_n$ into (\ref{eqSecMoCumu}) and comparing with (\ref{eqCumu2}), we notice that $\beta_{m,n}$ is entirely determined by the summation over non-crossing annular permutation and (\ref{eqSecMoCumu}) is only valid when the second order cumulant $c_{m,n}$ is zero. In summary, the cumulant sequences are
\begin{equation}\label{eqkpu}
c_n = \zeta\quad \mbox{and}\quad c_{m,n}=0,\quad m,n\ge 1.
\end{equation}

Apply the functional relation \cite[Eq. (52)]{CMSS07}
\begin{equation}\label{eqMxy}
M(x,y)=C\left(xM(x),yM(y)\right)\frac{\frac{d}{dx}\left(xM(x)\right)}{M(x)}\frac{\frac{d}{dy}\left(yM(y)\right)}{M(y)}+xy\left(\frac{\frac{d}{dx}\left(xM(x)\right)\,\frac{d}{dy}\left(yM(y)\right)}{(xM(x)-yM(y))^2}-\frac{1}{(x-y)^2}\right),
\end{equation}
where
\begin{align}
M(x) = 1 + \sum_{n\ge 1}\beta_n x^n,\quad M(x,y)=\sum_{m,n\ge1}\beta_{m,n} x^m y^n,\quad C(x,y) = \sum_{m,n\ge 1}c_{m,n}x^m y^n.\label{eqMxyb}
\end{align}
Due to (\ref{eqkpu}), the formal power series $C(x,y)=0$ and $x M(x) = \mc{G}_\mb{P}(1/x)$, and (\ref{eqMxy}) becomes
\begin{equation}
M(x,y) = \frac{\mc{G}_{\mb{P}}'(1/x)\mc{G}_{\mb{P}}'(1/y)}{x y (\mc{G}_{\mb{P}}(1/x)-\mc{G}_{\mb{P}}(1/y))^2}-\frac{xy}{(x-y)^2}.
\end{equation}
Comparing (\ref{eqRxy}) with (\ref{eqMxyb}), we obtain $\mc{R}(x,y) = M(x,y)/xy$. This completes the proof of Proposition~\ref{propCauchy2}.

\section{Asymptotic Capacity Distribution}\label{secCLT}

In this section, we prove a central limit theorem for the linear spectral statistics of the matrix $\mb{Q} = \mb{H}\mb{H}^{\dag}$ and show that the CDF of channel capacity $\mc{I}$ is asymptotically Gaussian as the matrix dimensions grow to infinity. This result generalizes the well-known CLT for the correlated Wishart matrix~\cite{Bai2004}. Together with the asymptotic mean and variance of capacity calculated in Propositions~\ref{propMu} and \ref{propVar}, the Gaussian convergence of capacity $\mc{I}$ gives a compact yet accurate approximation for the outage capacity. In addition, the approximative CDF of $\mc{I}$ is useful to analyze the diversity-multiplexing tradeoff of Rayleigh product channels in the finite SNR regime.

\subsection{Central Limit Theorem of Linear Spectral Statistics and Outage Probability}

Let $H_R(x) = R\left(\widetilde{F_{\mb{Q}}}(x) - F_{\mb{Q}}(x)\right)$, we are interested in the distribution of random variable
\begin{equation}
\mc{I} - \mu_{\mc{I}} = \int_{\mc{S}_{\mb{Q}}} \varphi(x)\mathrm{d}H_R(x).\label{eqmuH}
\end{equation}
Using the integral identities~(\ref{eqContourInt}) and (\ref{eqMIb}), we can rewrite (\ref{eqmuH}) as
\begin{equation}\label{eqmuG}
\mc{I} - \mu_{\mc{I}} = \frac{1}{2\pi \imath}\oint_{\mc{C}} \varphi(z)G_R(z)\mathrm{d}z,
\end{equation}
where $\mc{C}$ is the closed contour selected as in (\ref{eqContourInt}). In the following proposition, we prove that $G_R(z_1)$ and $G_R(z_2)$ with $z_1,z_2\in\mc{C}$ are jointly Gaussian distributed in the asymptotic regime~(\ref{eqDim}).

\begin{proposition}\label{propCLT}
In the asymptotic regime (\ref{eqDim}), $\{G_R(z)\}_{z\in\mc{C}}$ forms a tight sequence (see, e.g. \cite{Bai2004}) on a closed contour~$\mc{C}$ enclosing the support of $F_{\mb{Q}}(\cdot)$, and $G_{R}(z)$ converges weakly to a Gaussian process on the complex plane.
\end{proposition}
\begin{proof}
The proof of Proposition~\ref{propCLT} is in Appendix~\ref{appxCLT}.
\end{proof}
By Proposition~\ref{propCLT}, the asymptotic Gaussianity of (\ref{eqmuG}) follows from the fact that the Riemann sum corresponding to this integral has jointly Gaussian summands and the sum of which can only be Gaussian. Proposition~\ref{propCLT} generalizes the CLT of LSS for Wishart type random matrices involving one deterministic correlation matrix and one random matrix with \emph{i.i.d.} entries~\cite{Bai2004}. When both matrices $\mb{\Psi}$ and $\mb{\Theta}$ are random and independent, $G_R(z)$ can be decomposed into two random processes, see (\ref{eqGR}) in Appendix~\ref{appxCLT}. Both random processes are asymptotically Gaussian and each is governed by Lemma~\ref{lemmaCovBai}. As already discussed in Section~\ref{secVar}, the induced fluctuation of LSS is characterized by both the first order Cauchy transform and the second order \mbox{$R$-transform}. This is different from the Wishart random matrices, where the corresponding $G_R(z)$ only involves one asymptotic Gaussian process and the fluctuation of LSS is solely determined by the first order Cauchy transform. This makes the CLT of Rayleigh product ensembles distinct from the one in~\cite{Bai2004}. Together with the mean $\mu_{\mc{I}}$ and variance $\sigma_{\mc{I}}^2$ in Propositions~\ref{propMu} and~\ref{propVar}, an analytical Gaussian approximation to the capacity distribution of the Rayleigh product channel is available. This result can not be directly derived based on the existing results in~\cite{KamathTIT2005,HochwaldTIT2004,Tulino05,DM05,MSS03}. Note that the CLT of LSS for the biorthogonal ensembles, such as the Rayleigh product ensemble, was proved by Breuer and Duits~\cite{Breuer2013} for polynomial functions $\varphi(x)$. However, it is not clear how to extend this result to generic analytic functions $\varphi(x)$ such as the channel capacity $\varphi(x)=\log(1+\gamma x)$. Recently, the CLT for the product of two real and square random matrices was proved by G\"otze, Naumov, and Tikhomirov~\cite{GotzeProductReal2014} for smooth function $\varphi(x)$.

\begin{figure}[!t]
\centering
\includegraphics[width=4in]{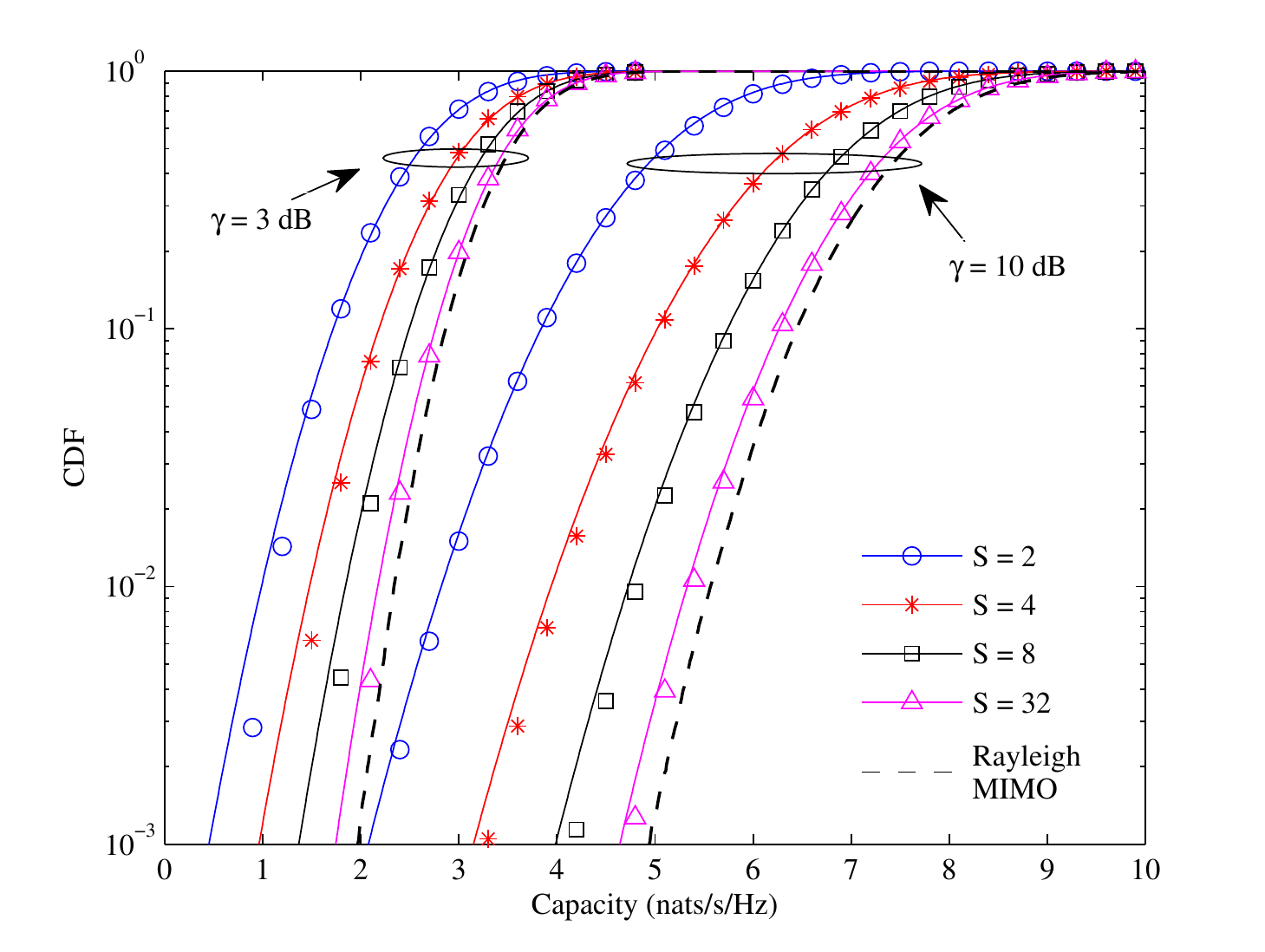}
\caption{CDF of channel capacity in the presence of $4\times 4$ Rayleigh product channel. Solid line: Gaussian approximation (\ref{eqPoutG}); markers: simulations; dashed line: CDF of $4\times 4$ MIMO capacity with independent Rayleigh fading.} \label{figPout}
\end{figure}

Let $\mathrm{erf}(x) = 2/\sqrt{\pi}\int_0^x e^{-t^2}\mathrm{d}t$ denote the error function, the Gaussian approximation to the CDF of channel capacity $\mc{I}$ is
\begin{equation}\label{eqPoutG}
F_{\mc{I}}(x) \approx \frac{1}{2}\left(1+\mathrm{erf}\left(\frac{x-\mu_{\mc{I}}}{\sigma_{\mc{I}}\sqrt{2}}\right)\right)
\end{equation}
and thus the outage capacity is
\begin{equation}\label{eqCoutG}
\mc{I}_{\mathrm{out}} \approx \mu_{\mc{I}} + \sigma_{\mc{I}}\sqrt{2}\ \mathrm{erf}^{-1}(2P_{\mathrm{out}}-1),\quad P_{\mathrm{out}}\in(0,1).
\end{equation}
Based on (\ref{eqPoutG}) and (\ref{eqCoutG}), the outage behavior of the Rayleigh product channel can be understood. In Fig.~\ref{figPout} the impact of the number of scatterers~$S$ as well as the received SNR $\gamma$ is studied, where a $4\times 4$ MIMO system is considered in the presence of $S = 2, 4, 8,\mbox{and } 32$ scattering objects. We plot the approximative outage probability as well as the empirical one obtained by Monte Carlo simulations. The outage probabilities are evaluated at SNRs $\gamma = 3$ and $\gamma = 10$~dB. As a comparison, we also plot the outage probability of a $4\times 4$ Rayleigh MIMO channel with independent fading entries. As the number of scatterers~$S$ increases, the outage capacity at a given probability level rapidly increases until $S$ is equal to the number of antennas, which is especially visible when SNR is large. In this range, the rank of the channel matrix is limited by the number of scatterers and increasing the scatterers effectively improves the rank of channel matrix. When $S>4$, the matrix rank is limited by the number of antennas and the improvement of outage capacity is relatively slow. Yet, the outage probability curve approaches to a limit, which corresponds to the outage probability of a conventional Rayleigh MIMO channel as predicted by~\cite{LL11}.

In Fig.~\ref{figCout} we examine the impact of number of antennas on the outage capacity. We plot the approximative $1\%$ outage capacity (\ref{eqCoutG}) as a function of received SNR $\gamma$. Assume the number of transmit and receive antennas $T = R = 2, 4, 8,\mbox{and }16$ while fixing the number of scatterers $S=8$. As expected, it is seen that the outage capacity of the Rayleigh product channel is lower than the conventional Rayleigh MIMO channel due to the presence of a finite number of scatterers. In the high SNR regime, the outage capacity curves of both channels attain the same slope when $T\le S$, which suggests that the capacity scales at the same rate as the limiting Rayleigh MIMO channel. On the other hand, when $S< T$ there is an increasing gap between the two channels as $\gamma$ increases. Finally, it is observed from Fig.~\ref{figPout} and~\ref{figCout} that the Gaussian approximation~(\ref{eqPoutG}) and (\ref{eqCoutG}) is reasonably accurate for a wide range of parameter settings.

\begin{figure}[!t]
\centering
\includegraphics[width=4in]{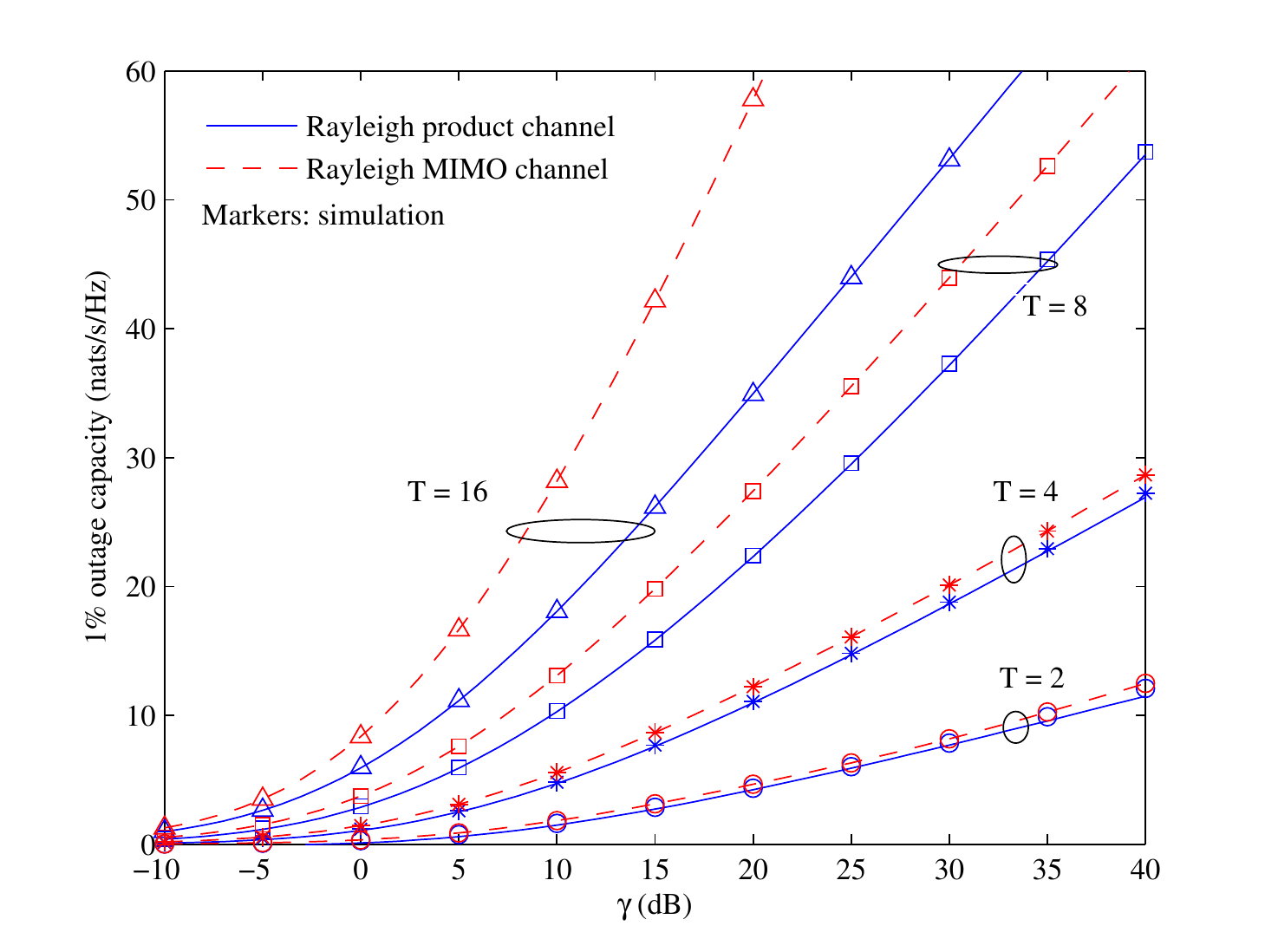}
\caption{$1\%$ outage capacity of Rayleigh product channels with $S=8$ scattering objects and equal numbers of antennas. Solid line: Gaussian quantile approximation~(\ref{eqCoutG}); markers: simulations; dashed line: outage capacity of conventional Rayleigh MIMO channels with equal number of antennas.} \label{figCout}
\end{figure}

\subsection{Finite-SNR Diversity-Multiplexing Tradeoff}

The concept of DMT was originally proposed in~\cite{ZhengTIT03} to characterize the diversity gain, which is related to link reliability, and the multiplexing gain, which is related to spectral efficiency. The DMT indicates that both types of performance gains can be obtained simultaneously while satisfying a fundamental tradeoff. The operational interpretation of the DMT framework is via the existence of universal codes, which are tradeoff optimal in the high SNR regime~\cite{TavildarTIT2006}.  In space-time code design~\cite{EliaTIT2006}, DMT represents a useful analytical tool to characterize the asymptotic performance of codes. However, the asymptotic tradeoff is a too optimistic upper bound to estimate the operational performance at realistic SNRs. Recent works have shown that codes optimized at high SNR may not be optimal at low or moderate SNR. Motivated by these facts, Narasimhan~\cite{NarasimhanTIT2006} proposed a finite-SNR DMT framework, which characterizes the non-asymptotic DMT. There, he studied the finite DMT for the correlated Rayleigh and Rician MIMO channels at realistic SNR levels.

Under the assumptions of slow fading and capacity achieving codes with rate $r$, the multiplexing gain $m$ of a MIMO channel is defined according to~\cite[Eq. (21)]{LoykaTIT2010} as
\begin{equation*}\label{eqMul}
m = \frac{n}{\mu_{\mc{I}}}r,
\end{equation*}
where $n = \min(R,S,T)$. The multiplexing gain  provides an indication of the sensitivity of rate adaptation strategy as the SNR changes. When the applied codes have a higher multiplexing gain, the rate adaptation tends to respond more dramatically to the SNR variations. At a fixed multiplexing gain, the finite-SNR diversity gain $d(m,\gamma)$ is defined as the negative slope of the log-log plot of outage probability $P_{\mathrm{out}}(r)$ at rate $r=m\,\mu_{\mc{I}}/n$ versus SNR $\gamma$,
\begin{equation}\label{eqDiv}
d(m,\gamma) = -\frac{\partial \log P_{\mathrm{out}}\left(m\,\mu_{\mc{I}}/n\right)}{\partial \log\gamma}.
\end{equation}
At a particular SNR $\gamma$ and multiplexing gain $m$, the diversity gain (\ref{eqDiv}) provides an estimate of the additional SNR needed to reduce the outage probability by a certain amount. Using the derived outage probability (\ref{eqPoutG}), the finite-SNR DMT can be obtained for the Rayleigh product channel.
\begin{proposition}\label{propDMT}
When $R = T$, the finite-SNR DMT of Rayleigh product channels can be approximated by
\begin{equation}\label{eqDivApprox}
d(m,\gamma) = \frac{2\gamma}{\sqrt{\pi}}\frac{\exp(-K(m,\gamma)^2/2)}{1+\mathrm{erf}(-K(m,\gamma))}\frac{\partial K(m,\gamma)}{\partial \gamma},
\end{equation}
where $K(m,\gamma) = \frac{n-m}{\sqrt{2}n}\frac{\mu_{\mc{I}}}{\sigma_{\mc{I}}}$ with $\mu_{\mc{I}}$ and $\sigma_{\mc{I}}^2$ calculated by Propositions~\ref{propMu} and \ref{propVar}.
\end{proposition}
\begin{proof}
The proof of Proposition~\ref{propDMT} follows by substituting (\ref{eqPoutG}) into (\ref{eqDiv}).
\end{proof}
Note that the approximation (\ref{eqDivApprox}) is tight in the asymptotic regime~(\ref{eqDim}). This is because the approximation error is induced from (\ref{eqPoutG}).

Fig.~\ref{figDMTT2S2} shows the finite-SNR DMT of a $2\times 2$ Rayleigh product channel with  $S=2$ scatterers. The approximated tradeoff curves are generated by (\ref{eqDivApprox}) at SNRs $\gamma = 0$ dB and $\gamma = 5$ dB. Compared to the Monte Carlo simulations, the proposed approximation yield close estimate for the MIMO diversity gain. As $m$ approaches the maximum multiplexing gain, the discrepancies between the approximation and simulation curves decrease. When $R = T = 4$ antennas are used, the MIMO channel achieves improved channel diversity for a given multiplexing gain as shown in Fig.~\ref{figDMTT4S2}. In both figures, we have also plotted the asymptotic DMT of Rayleigh product channels according to~\cite[Eq. (8)]{Yang11}, when SNR $\gamma$ approaches infinity. It is clear that the asymptotic results significantly overestimate the channel diversity at the considered operational SNR levels, which justifies the usefulness of the proposed approximation~(\ref{eqDivApprox}).

\begin{figure}[!t]
\centerline{\includegraphics[width = 4in]{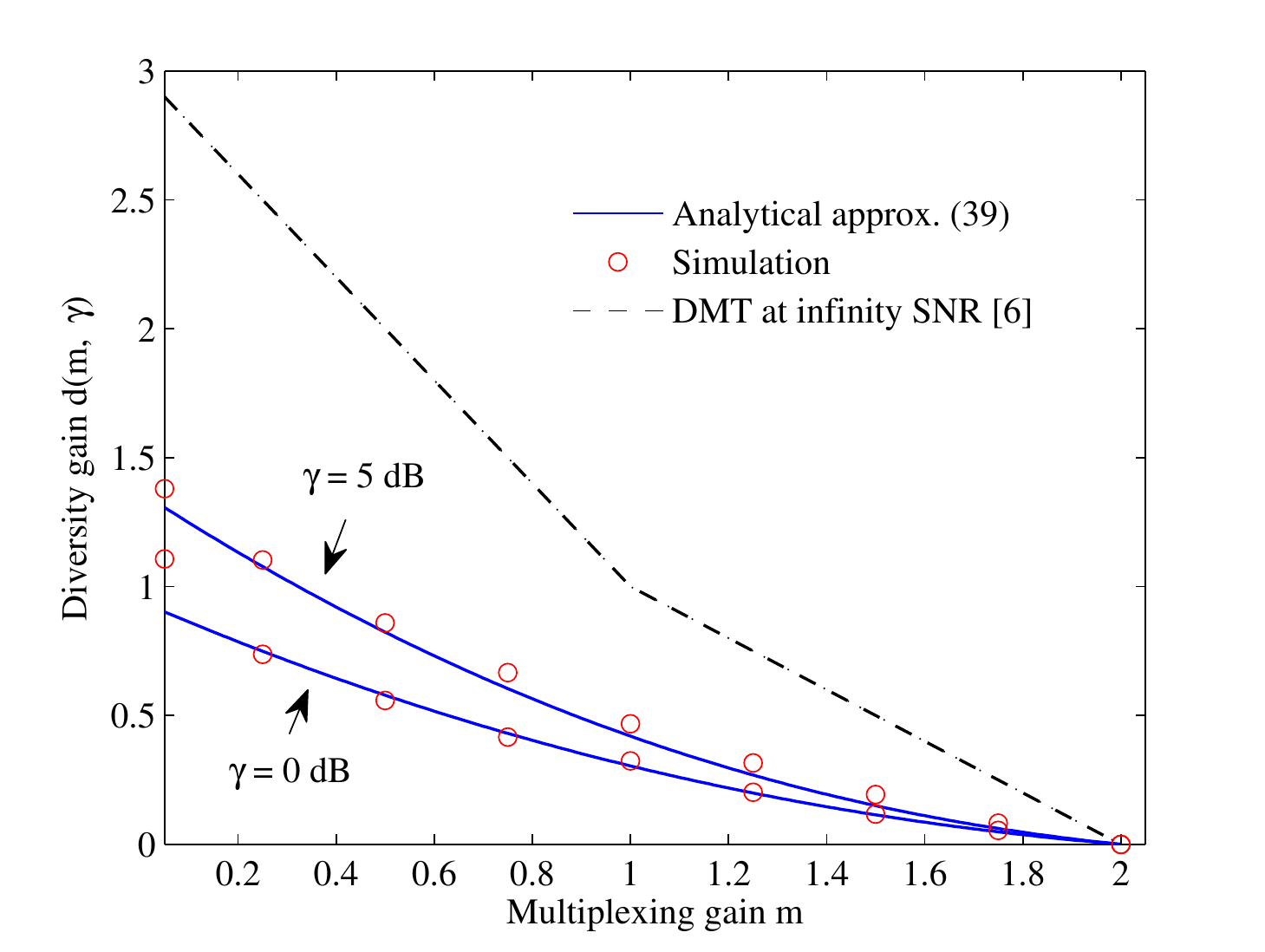}}
\caption{Finite-SNR DMT of $2\times 2$ Rayleigh product channel with $S=2$. Solid line: approximation calculated by (\ref{eqDivApprox}); markers: simulations; dashed line: asymptotic DMT with SNR $\gamma\rightarrow\infty$.}\label{figDMTT2S2}
\end{figure}

\begin{figure}[!t]
\centerline{\includegraphics[width= 4in]{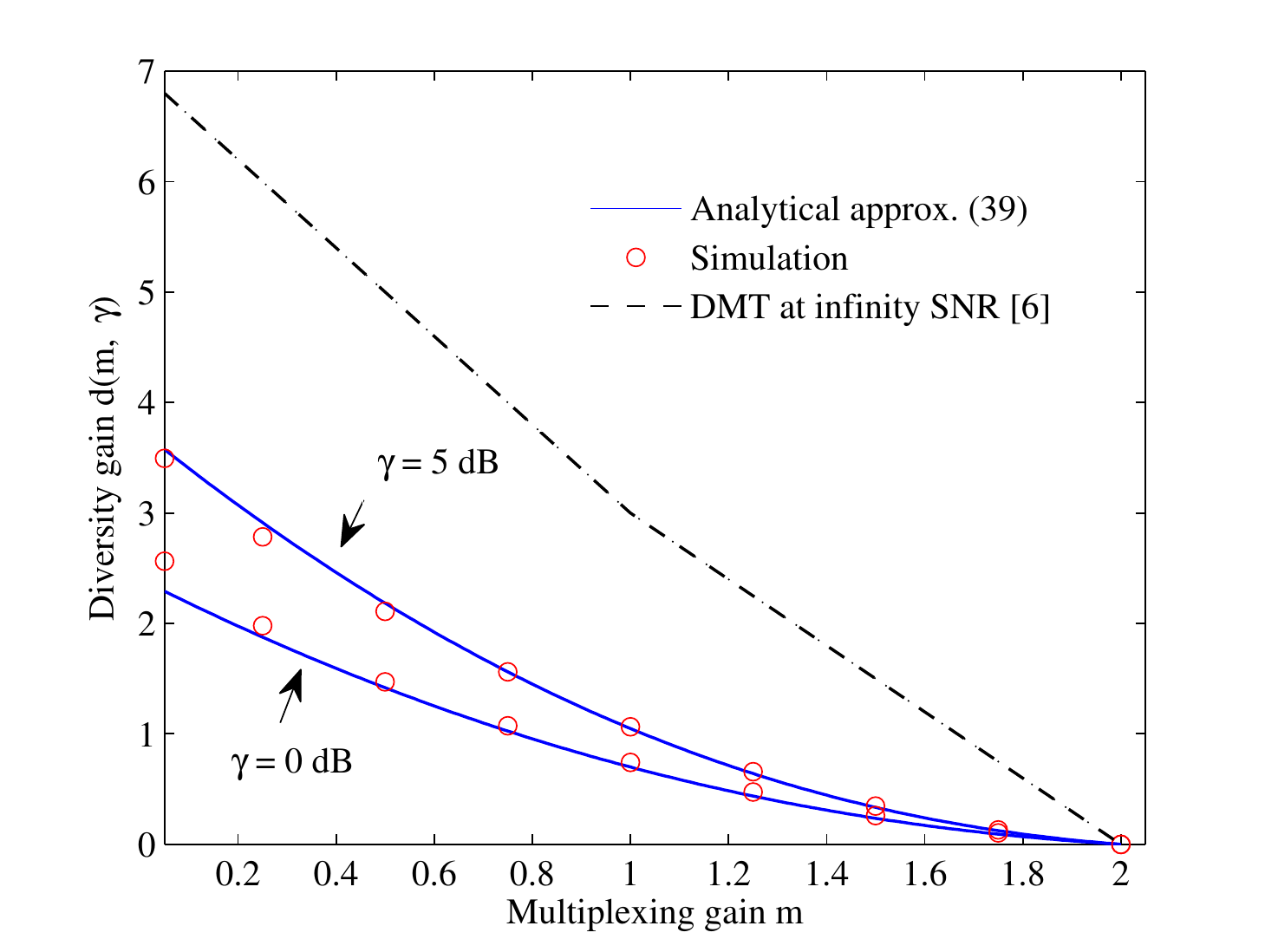}}
\caption{Finite-SNR DMT of $4\times 4$ Rayleigh product channel with $S=2$. Solid line: approximation calculated by (\ref{eqDivApprox}); markers: simulations; dashed line: asymptotic DMT with SNR $\gamma\rightarrow\infty$.}\label{figDMTT4S2}
\end{figure}

\section{Conclusions}\label{secConc}
We studied the outage probability of Rayleigh product channels, which explicitly model the rank deficiency effect. Using free probability theory, the asymptotic variance of channel capacity is calculated for large channel matrix and becomes exact when matrix dimensions approach infinity. Compared to the conventional Rayleigh MIMO channels, the Rayleigh product channels induce a higher capacity fluctuation, which is determined by the second order \mbox{$R$-transform} of the channel matrix. We have proved that the channel capacity is asymptotically Gaussian by establishing a CLT of a relevant linear spectral statistics.  Numerical results show that the proposed Gaussian approximation is reasonably accurate for realistic channel dimensions. Results have been utilized to characterize the tradeoff between diversity and multiplexing of Rayleigh product channels, while the asymptotic tradeoff for large SNR may be an over-optimistic estimate.

\section*{Acknowledgment}
Z. Zheng and J. H\"am\"al\"ainen are supported by the Academy of Finland (Grant 284634). L. Wei and J. Corander are supported by the Academy of Finland (Grant 251170). R. Speicher is supported by the ERC Advanced Grant (NCDFP 339760). R. M\"uller is supported by the Alexander von Humboldt Foundation. We would like to thank O. Arizmendi and A. Soshnikov for helpful discussions.

\begin{appendices}

\section{Proof of Proposition~\ref{propVar}}\label{appxAsyVar}

Let $\mc{S}_{\mb{Q}}^c\in\mbb{R}$ denote the complement of the support $\mc{S}_{\mb{Q}}$ on the real axis. It is shown in \cite{Silver1995} that for a given open interval $\mc{T}\subset\mc{S}_{\mb{Q}}^c$, the function $\mc{G}_{\mb{Q}}(\cdot)$ is continuous, real, and decreasing. This is also true for $\mc{G}_{\mb{P}}(\cdot)$ with $\mc{T}\subset\mc{S}_{\mb{P}}^c$. Therefore, there exists an inverse function $G^{-1}(t)=\mc{G}_{\mb{Q}}^{-1}\left(1/\mc{G}_{\mb{P}}^{-1}(t)\right)$ continuous, real, and decreasing over $\{t\in\mbb{R}: t = G(x), x\in\mc{S}_{G}^c\}$. We choose the contour $\mc{C}_x$ to be inside of $\mc{C}_y$ such that they both cross real-axis in the intervals $(-1/\gamma,0)$ and $(\lambda_r,\infty)$, where $\lambda_r$ denotes the right end-point of the support $\mc{S}_{\mb{Q}}$. By substitutions $t_1 = G(x)$ and $t_2 = G(y)$, the integral (\ref{eqVarc}) can be alternatively integrated over contours $\mc{C}_1$ and $\mc{C}_2$ as
\begin{equation}
\sigma_{\mc{I}}^2 = -\frac{1}{4\pi^2}\oiint_{\mc{C}_1,\mc{C}_2} \frac{\varphi\left(G^{-1}(t_1)\right)\varphi\left(G^{-1}(t_2)\right)}{(t_1-t_2)^2}\mathrm{d}t_1\mathrm{d}t_2 = \frac{1}{2\pi \imath}\oint_{\mc{C}_2}\varphi\left(G^{-1}(t_2)\right)\mc{K}_{\mathrm{inner}}(t_2)\mathrm{d}t_2,\label{eqVard}
\end{equation}
where
\begin{equation}
\mc{K}_\mathrm{inner}(t_2) = \frac{1}{2\pi \imath}\oint_{\mc{C}_1} \frac{\log(1+\gamma G^{-1}(t))}{(t-t_2)^2}\,\mathrm{d}t.\label{eqKinner}
\end{equation}
The transformed contours $\mc{C}_1$ and $\mc{C}_2$ cross the real-axis in the intervals $(G(0^-), G(-1/\gamma))=(-\infty, G(-1/\gamma))$ and $(G(\infty), G(\lambda_r)) = (0, G(\lambda_r))$, where
\begin{equation*}
G(0^-)=\lim_{x\rightarrow 0-} G(x), \quad G(\infty) = \lim_{x\rightarrow\infty} G(x).
\end{equation*}
The inverse function $\mc{G}_{\mb{P}}^{-1}(t)$ is calculated via (\ref{eqGp}) as
\begin{equation}
\mc{G}_{\mb{P}}^{-1}(t) = \frac{1}{t} + \frac{\zeta}{1-t}.\label{eqGpinv}
\end{equation}
We obtain the inverse function $\mc{G}_{\mb{Q}}^{-1}(t)$ by solving the quadratic equation (\ref{eqcubicGQ}) in $z$ as
\begin{equation}
\mc{G}_{\mb{Q}}^{-1}(t) = \frac{1-(1-\zeta)t\pm \sqrt{1+(1-\zeta)^2 t^2-2(1+\zeta)t}}{2\zeta t^2},\label{eqGqinv}
\end{equation}
where the minus sign is taken by the requirement $\lim\limits_{z\rightarrow\infty} z\mc{G}_{\mb{Q}}(z) = 1$~\cite{Silver1995}. Substituting (\ref{eqGpinv}) and (\ref{eqGqinv}) into~(\ref{eqKinner}) and applying integration by parts, we can rewrite $\mc{K}_{\mathrm{inner}}(t_2)$ as
\begin{align}
\mc{K}_\mathrm{inner}(t_2) &= \frac{1}{2\pi \imath}\oint_{\mc{C}_1} \frac{\gamma \left(G^{-1}(t)\right)'}{(t-t_2)\left(1+\gamma G^{-1}(t)\right)}\,\mathrm{d}t\nonumber\\
&= -\frac{\gamma}{2\pi \imath}\oint_{\mc{C}_1}\frac{2(\zeta-1)t^2+3t-1}{t(t-1)(t-t_2)(t-\omega_r)(t-\omega_{+})(t-\omega_{-})}\,\mathrm{d}t,\label{eqInner}
\end{align}
where $\omega_r$, $\omega_{+}$, and $\omega_{-}$ are the three roots of the cubic equation $t^3-2t^2+(1-\gamma+\gamma\zeta)t+\gamma=0$ and~$\omega_r$ denotes the real solution such that $\omega_r = G(-1/\gamma)<0$. The integrand of (\ref{eqInner}) has two simple poles at $t=0$ and $t=\omega_r$ in $\mc{C}_1$, and by applying the residue theorem, the integral $\mc{K}_{\mathrm{inner}}(t_2)$ becomes
\begin{equation}
\mc{K}_{\mathrm{inner}}(t_2) = \frac{1}{t_2}-\frac{1}{t_2-\omega_r}.\label{eqInnerb}
\end{equation}

Substituting (\ref{eqInnerb}) into (\ref{eqVard}), the variance $\sigma_\mc{I}^2$ can be therefore expressed as
\begin{align}
\sigma_\mc{I}^2 &= \frac{1}{2\pi \imath}\oint_{\mc{C}_2} \log\left(1+\gamma G^{-1}(t)\right)\left(\frac{1}{t}-\frac{1}{t-\omega_r}\right)\,\mathrm{d}t\nonumber\\
&= \frac{1}{2\pi \imath}\oint_{\mc{C}_2} \log\frac{(t-\omega_{+})(t-\omega_{-})}{(t-1)^2}\left(\frac{1}{t}-\frac{1}{t-\omega_r}\right)\,\mathrm{d}t + \frac{1}{2\pi \imath}\oint_{\mc{C}_2} \log\frac{t-\omega_r}{t}\left(\frac{1}{t}-\frac{1}{t-\omega_r}\right)\,\mathrm{d}t.\label{eqSigd}
\end{align}
The second integral in (\ref{eqSigd}) has an anti-derivative $\left(\log\frac{t-\omega_r}{t}\right)^2/2$, which is single-valued over $\mc{C}_2$ and therefore vanishes due to Cauchy's theorem. Applying the residue theorem to the first integral in (\ref{eqSigd}), we obtain
\begin{equation*}
\sigma_{\mathcal{I}}^2 = \log\frac{(\omega_r-1)^2\omega_{+}\omega_{-}}{(\omega_r-\omega_{+})(\omega_r-\omega_{-})}.
\end{equation*}
The proof is completed by the fact that $\omega_r\omega_{+}\omega_{-} = -\gamma$.

\section{Proof of Corollary~\ref{corollary1}}\label{appxCorollary1}
When $\zeta>1$, $u(\gamma,\zeta)$ and $u(\gamma,\zeta)^2$ can be expanded at $\gamma\rightarrow\infty$ as
\begin{align*}
u(\gamma,\zeta) &= \sqrt{3(\zeta-1)}\gamma^{1/2}-\frac{2\zeta+1}{2(\zeta-1)} + \frac{\sqrt{3}(4\zeta-1)}{8(\zeta-1)^{5/2}}\gamma^{-1/2} + \mc{O}(\gamma^{-1}),\\
u(\gamma,\zeta)^2 &= 3(\zeta-1)\gamma-\frac{\sqrt{3}(2\zeta+1)}{(\zeta-1)^{1/2}}\gamma^{1/2} + \mc{O}(1).
\end{align*}
The real and negative solution $\omega_r$ of (\ref{eqCubic}) corresponds to $t_1$ and it follows from (\ref{eqt1}) that
\begin{align}\label{eqwrr}
\omega_r = \frac{2}{3}-\frac{3\gamma\zeta-3\gamma-1-u(\gamma,\zeta)^2}{3u(\gamma,\zeta)} = \frac{1}{1-\zeta} + \mc{O}(\gamma^{-1/2}).
\end{align}

When $\zeta<1$, the asymptotic expansion of $u(\gamma,\zeta)$ at $\gamma\rightarrow\infty$ yields
\begin{equation}\label{equl}
u(\gamma,\zeta) = (-1)^{1/6}\sqrt{3(1-\zeta)}\gamma^{1/2} + \mc{O}(1) = e^{\frac{\imath\pi}{6}} \sqrt{3(1-\zeta)}\gamma^{1/2} + \mc{O}(1),
\end{equation}
where we took the principle value $(-1)^{1/6} = e^{\imath\pi/6}$. In this case, the real and negative solution $\omega_r$ of~(\ref{eqCubic}) corresponds to $t_2$. Inserting (\ref{equl}) into (\ref{eqt2}), we have
\begin{equation}\label{eqwrl}
\omega_r = \frac{2}{3}-\frac{1}{3}\sqrt{3(1-\zeta)\gamma}\left(e^{\frac{\imath\pi}{6}}+e^{-\frac{\imath\pi}{6}}\right) = \frac{2}{3}-\sqrt{(1-\zeta)\gamma}.
\end{equation}
Inserting (\ref{eqwrr}) into (\ref{eqVars}) and taking derivative with respect to $\zeta$, we obtain
\begin{align*}
\frac{\partial}{\partial \zeta}\sigma_{\mc{I}}^2 \approx - \frac{2\left((\zeta-1)^3\gamma-2\zeta^2+\zeta\right)}{\zeta(\zeta-1)\left((\zeta-1)^3 \gamma +2\zeta\right)}\approx -\frac{2}{\zeta(\zeta-1)} < 0.
\end{align*}
Therefore, when $\zeta>1$ the asymptotic variance $\sigma_{\mc{I}}^2$ is a monotonically decreasing function of $\zeta$. When $\zeta<1$, the variance $\sigma_{\mc{I}}^2$ is a monotonically increasing function of $\zeta$, where the derivative $\partial \sigma_{\mc{I}}^2/\partial \zeta > 0$ with $\omega_r$ given by (\ref{eqwrl}). To sum up, the asymptotic variance is maximum when $\zeta$ approach $1$ from both sides of the axis. This completes the proof of Corollary~\ref{corollary1}.

\section{Non-Crossing Permutations}\label{appxNC}
Let us introduce the main combinatorial objects, the non-crossing disc and annular permutations, and the related notations, which are used in Lemma~\ref{lemmaCumu} and in the proof of Proposition~\ref{propCauchy2}. We refer the readers to~\cite{NS06,MN04,SpeicherICM14} for a comprehensive description of the non-crossing permutations.

For a positive integer $n$, we denote the set $\{1,\ldots,n\}$ as $[n]$. Let $\mc{P}_n$ denote the set of all partitions of $[n]$. Given a partition $\pi\in\mc{P}_n$, we have $\pi=\{B_1,\ldots,B_k\}$, where $B_1,\ldots,B_k$, called blocks of $\pi$, are non-empty disjoint subset of $[n]$, i.e. $B_1\cup\cdots\cup B_k=[n], \mbox{ and }B_i \cap B_j=\emptyset\ \mbox{for}\ i\neq j$.
Given two partitions $\pi_1,\pi_2\in\mc{P}_n$, we have $\pi_1\le\pi_2$ if and only if every block of $\pi_1$ is contained in a block of $\pi_2$ and denote $1_{n} = \{1,\ldots,n\}$ the largest partition over  $[n]$. We say a partition $\pi$ is non-crossing in disc sense if there does not exist $1\le i,j\le k$, $i\neq j$, and $1\le a<b<c<d\le n$, such that $a,c\in B_i$ and $b,d\in B_j$. A non-crossing disc partition $\pi\in\mc{P}_n$ can be visualized as follows: draw the points $1,\ldots,n$ clockwise around the boundary of a disc and connect the points belonging to the same block with a convex hull. The partition $\pi$ is non-crossing if the convex hulls are pairwise disjoint.

A concept closely related to the partition is the set permutation. Let $\mc{S}_n$ denote the set of all permutations over $[n]$. Given a permutation $\tau\in\mc{S}_n$, we have $\tau = A_1\cdots A_k$ such that $[n]$ is decomposed into $k$ orbits and $A_i = (a_i(1),\ldots,a_i(s))$ is the $i$-th orbit of $\tau$ containing $s$ elements. For two elements $a_i(p),a_i(q)$ belong to the same orbit $A_i$, there exists an integer $m\ge 1$ such that $\tau^m(a_i(p))=a_i(q)$. For instance, if $\tau=(1,4,5)(2,3)\in\mc{S}_5$, it maps the elements as $\tau(1) = 4$, $\tau(4)=5$, $\tau(5)=1$, $\tau(2)=3$, and $\tau(3)=2$. The notion $\#(\tau)$ is used as the number of orbits of $\tau$. We say a permutation is standard in disc sense if for every orbit $A_i=(a_i(1),\ldots,a_i(s))$ of $\tau$, there is $a_i(1)<\cdots<a_i(s)$. A standard disc permutation~$\tau$ has an induced partition $\pi$, where each block of $\pi$ contains the same elements as the corresponding orbit of $\tau$. In addition, if the partition induced by standard permutation $\tau$ is non-crossing in disc sense, $\tau$ is a non-crossing disc permutation and we denote the set of all non-crossing disc permutation on $[n]$ as~$\mc{S}_{\operatorname{d-nc}}(n)$. Let $\eta = (1,\ldots,n)$ the forward cyclic permutation of $[n]$. A permutation $\tau\in\mc{S}_n$ satisfy the so-called geodesic condition as
\begin{equation}\label{eqDiscGeo}
\#(\tau) + \#(\tau^{-1}\eta) \le n+1,
\end{equation}
where $\tau^{-1}\eta$ or alternative $\tau^{-1}\circ\eta$ is the composite permutation by first applying $\eta$ and then $\tau^{-1}$. The equality in (\ref{eqDiscGeo}) only holds when $\tau$ is non-crossing disc permutation. The geodesic condition can be intuitively viewed as the triangular inequality for the Cayley graph of permutation group $\mc{S}_n$~\cite{Biane1997}. Let $\tau_1,\tau_2\in\mc{S}_n$, the distance between $\tau_1$ and $\tau_2$ in Cayley graph of $\mc{S}_n$ amounts to $d(\tau_1,\tau_2) = n-\#(\tau_1^{-1}\tau_2)$. The inequality (\ref{eqDiscGeo}) can be rewritten in terms of Cayley distance as $d(id,\tau) + d(\tau,\eta) \ge d(id,\eta)$, where $id$ is the identity permutation. The condition that permutation $\tau$ is non-crossing is equivalent that $\tau$ lies on the geodesic connecting $id$ and $\eta$ in the Cayley graph.

Let us consider another set of permutations $\mc{S}_{m+n}$, illustrated via topological drawing in the $(m,n)$-annular sense. Instead of placing $m+n$ points on the boundary of one disc, we will use two concentric circles. The points $1,\ldots,m$ are placed clockwise on the external circle and the points $m+1,\ldots,m+n$ are placed counter-clockwise on the internal circle. The annulus between the two circles are referred to as $(m,n)$-annulus. Given a permutation $\tau\in\mc{S}_{m+n}$, it is visualized by drawing curves within the $(m,n)$-annulus, which connect the elements of each orbit, respectively.  Let $A=(a(1),\ldots,a(s))$ an orbit of $\tau$ with $s$ elements. The corresponding curve connects $a(1)$ to $a(2)$, then $a(2)$ to $a(3)$, \ldots, then $a(s)$ to $a(1)$ such that: 1) it does not intersect with itself; 2) it encloses a region completely contained in $(m,n)$-annulus; 3) it goes clockwise around the region. We say a permutation $\tau$ is $(m,n)$-connected if there is at least one orbit of $\tau$ contains elements on both circles, otherwise $\tau$ is $(m,n)$-disconnected. In addition, $\tau\in\mc{S}_{m+n}$ is standard in $(m,n)$-annular sense if each orbit of $\tau$ satisfies either of the following conditions:
\begin{enumerate}
\item Given an orbit $A\subseteq\tau$ such that $A\subseteq\{1,\ldots,m\}$ or $A\subseteq\{m+1,\ldots,m+n\}$. The elements of $A$, upon cyclic permutations, can be sorted in increasing order.
\item $A\cap\{1,\ldots,m\}\neq\emptyset$ and $A\cap\{m+1,\ldots,m+n\}\neq\emptyset$. We have $A=(a(1),\ldots,a(k),b(1),\ldots,b(l))$, where $a(1),\ldots,a(k)\in\{1,\ldots,m\}$ and $b(1),\ldots,b(l)\in\{m+1,\ldots,m+n\}$. Both sequences $\{a(i)\}$ and $\{b(j)\}$, upon cyclic permutations, can be sorted in increasing order, respectively.
\end{enumerate}
We say a permutation $\tau\in\mc{S}_{m+n}$ is non-crossing in $(m,n)$-annular sense if $\tau$ is standard and the regions enclosed by every orbits of $\tau$ are not overlapping in the annular visualization described above. We denote~$\mc{S}_{\operatorname{a-nc}}(m,n)$ as the set of non-crossing $(m,n)$-annular permutations. Finally, according to~\cite[Th. 6.1]{MN04}, a permutation $\tau\in\mc{S}_{m+n}$ and $(m,n)$-connected satisfies a geodesic condition in the $(m,n)$-annular sense as
\begin{equation}\label{eqAnnGeo}
\#(\tau) + \#(\tau^{-1}\eta_0)\le m+n,
\end{equation}
where $\eta_0 = (1,\ldots,m)(m+1,\ldots,m+n)$ and the equality only holds when $\tau\in\mc{S}_{\operatorname{a-nc}}(m,n)$.

\section{Proof of Lemma \ref{lemmaCumu}}\label{appxLemmaCumu}
The proof relies on a known combinatorial identity of the moments of Gaussian random variables, which is stated below.
\begin{lemma}(Wick's Lemma~\cite{MN04}).\label{lemmaWick}
Let $Z_1,\ldots,Z_t$ denote \emph{i.i.d.} complex Gaussian random variables with zero mean and unit variance.
\begin{enumerate}
\item Let $m,n$ be positive integers such that $m\neq n$, and consider two functions $\alpha: [m]\rightarrow [t]$ and $\beta:[n]\rightarrow [t]$. Then
\begin{equation*}
\mbb{E}\left[Z_{\alpha(1)}\cdots Z_{\alpha(m)}\overbar{Z}_{\beta(1)}\cdots\overbar{Z}_{\beta(n)}\right]=0.
\end{equation*}
\item Let $n$ be a positive integer and consider two functions $\alpha,\beta:[n]\rightarrow [t]$. Then
\begin{equation}\label{eqWick}
\mbb{E}\left[Z_{\alpha(1)}\cdots Z_{\alpha(n)}\overbar{Z}_{\beta(1)}\cdots\overbar{Z}_{\beta(n)}\right]=\mbox{card}\left\{\tau\in\mc{S}_n|\alpha=\beta\circ\tau\right\},
\end{equation}
where $\mbox{card}(\cdot)$ denotes the cardinality.
\end{enumerate}
\end{lemma}

Denote $\mb{Q}=(q_{ij})_{i,j=1}^{R}$ with the entry $q_{ij}$ given by
\begin{equation}\label{eqQij}
q_{ij} = \frac{1}{R S}\sum_{a=1}^{S}\sum_{b=1}^{T}\sum_{c=1}^{S} \bar{\psi}_{ai}\theta_{ab}\bar{\theta}_{cb}\psi_{cj}.
\end{equation}
In light of \cite[Th. 2.12]{CMSS07}, the second order free cumulant of $\mb{Q}$ can be expressed in terms of classic joint cumulants of entries $q_{ij}$ as
\begin{equation}\label{eqCumu2b}
\kappa_{m,n}=\lim_{R\rightarrow\infty}R^{m+n} k_{m+n}\left(\mb{q}_{m,n}\right),
\end{equation}
where the vector $\mb{q}_{m,n}=[q_{i(1)i(2)},q_{i(2)i(3)},\ldots,q_{i(m)i(1)},q_{i(m+1)i(m+2)},q_{i(m+2)i(m+3)},\ldots,q_{i(m+n)i(m+1)}]$ can be any distinct choice of $i(1),\ldots, i(m+n)$. For a partition $\pi\in\mc{P}_n$, we define $\pi_i=\{\pi_i(1),\ldots,\pi_i(s)\}\subset\pi$ as a block of $\pi$ with $s$ elements. The expectation over the blocks of partition $\pi$ is defined as
\begin{equation*}
\mbb{E}_{\pi}[a_1,\ldots,a_n] = \prod_{\pi_i\subset\pi} \mbb{E}\left[a_{\pi_i(1)}\cdots a_{\pi_i(s)}\right].
\end{equation*}
Using the cumulant-moment relations~\cite[Eq. (10)]{CMSS07}, $k_{m+n}$ can be written as a sum of $\mbb{E}_{\pi}[\mb{q}_{m,n}]$ for all $\pi\in\mc{P}_{m+n}$, namely
\begin{align}\label{eqkmn}
k_{m+n}(\mb{q}_{m,n}) = \sum_{\pi\in\mc{P}_{m+n}}\mbb{E}_{\pi}[\mb{q}_{m,n}]\,\mbox{M\"ob}_{\mc{P}_{m+n}}(\pi,1_{m+n}),
\end{align}
where $\mbox{M\"ob}_{\mc{P}_{m+n}}:\mc{P}_{m+n}\times\mc{P}_{m+n}\rightarrow\mbb{C}$ denotes the \mbox{M\"obius} function~\cite{NS06} on $\mc{P}_{m+n}$, which satisfies
\begin{align}\label{eqMobius}
\sum_{\genfrac{}{}{0pt}{}{\eta\in\mc{P}_{m+n}}{\pi\le\eta}} \mbox{M\"ob}_{\mc{P}_{m+n}}(\eta,1_{m+n}) =
	\begin{cases}
		1 & \mbox{if   } \pi = 1_{m+n}\\
		0 & \mbox{otherwise} .
	\end{cases}
\end{align}

Inserting (\ref{eqQij}) into (\ref{eqkmn}) and applying Lemma \ref{lemmaWick}, we see that for a given partition $\pi$ the multiplicative moment $\mbb{E}_{\pi}[\mb{q}_{m,n}]$ is non-zero only when the partition $\pi$ takes the forms $\pi^{(1)}=1_{m+n}$, $\pi^{(2)}=\{\{1,\ldots,m\},\{m+1,\ldots,m+n\}\}$. The corresponding \mbox{M\"obius} function  can be calculated via (\ref{eqMobius}) as
\begin{equation}\label{eqMobius2}
\mbox{M\"ob}_{\mc{P}_{m+n}}\left(\pi^{(1)},1_{m+n}\right)=1,\quad \mbox{M\"ob}_{\mc{P}_{m+n}}\left(\pi^{(2)},1_{m+n}\right)=-1.
\end{equation}
It follows from (\ref{eqkmn}) and (\ref{eqMobius2}) that the cumulant $k_{m+n}(\mb{q}_{m,n})$ equals $k_{m+n}(\mb{q}_{m,n})=\mbb{E}_{\pi^{(1)}}[\mb{q}_{m,n}]-\mbb{E}_{\pi^{(2)}}[\mb{q}_{m,n}]$.

Denote the permutation $\eta_0 = (1,\ldots,m)(m+1,\ldots,m+n)\in\mc{S}_{m+n}$. Substituting $q_{i(t)i(\eta_0(t))}$, $t=1,\ldots,m+n$, into (\ref{eqQij}), we can express $\mbb{E}_{\pi^{(1)}}[\mb{q}_{m,n}]$ as
\begin{align}\label{eqPi1}
\frac{1}{(R S)^{m+n}}\sum_{\substack{1\le a_1,\ldots,a_{m+n}\le S\\ 1\le b_1,\ldots,b_{m+n}\le T\\1\le c_1,\ldots,c_{m+n}\le S}}\mbb{E}\left[\prod_{t=1}^{m+n}\bar{\psi}_{a_t i(t)}\theta_{a_t b_t}\bar{\theta}_{c_t b_t}\psi_{c_t i(\eta_0(t))}\right].
\end{align}
It is convenient to introduce the functions $A:[m+n]\rightarrow [S]$, $B:[m+n]\rightarrow [T]$, $C:[m+n]\rightarrow [S]$. Due to the independence between $\psi_{ij}$ and $\theta_{kl}$, we can rewrite (\ref{eqPi1}) as
\begin{align}\label{eqPi1b}
\frac{1}{(R S)^{m+n}}&\sum_{A,B,C}\mbb{E}\left[\prod_{t=1}^{m+n}\theta_{A(t)B(t)}\prod_{t=1}^{m+n}\bar{\theta}_{C(t)B(t)}\right]\,\mbb{E}\left[\prod_{t=1}^{m+n} \psi_{C\left(\eta_0^{-1}(t)\right)i(t)}\prod_{t=1}^{m+n}\bar{\psi}_{A(t)i(t)}\right].
\end{align}
Since the indexes $i(1),\ldots,i(m+n)$ are distinct, by Lemma \ref{lemmaWick} the second expectation in (\ref{eqPi1b}) is non-zero only when $C\circ\eta_0^{-1}=A$. The first expectation in (\ref{eqPi1b}) is calculated by (\ref{eqWick}) with $\alpha(t) = (A(t),B(t))$ and $\beta(t) = (C(t), B(t))$. For a given permutation $\tau\in\mc{S}_{m+n}$, the summands of (\ref{eqPi1b}) should fulfill $A = C\circ\tau$ and $B = B\circ\tau$ to be able to contribute to the summation. To summerize, $\mbb{E}_{\pi^{(1)}}[\mb{q}_{m,n}]$ is expressed as
\begin{align}
\mbb{E}_{\pi^{(1)}}[\mb{q}_{m,n}]=\frac{1}{(R S)^{m+n}}\sum_{B,C}\mbox{card}\left\{\tau\in\mc{S}_{m+n}\mid C\circ(\tau^{-1}\eta_0)=C, B\circ\tau=B\right\}.\label{eqProdMoment1a}
\end{align}
Interchange summation and cardinality operations in (\ref{eqProdMoment1a}) and write (\ref{eqProdMoment1a}) as a sum over the permutation~$\tau$,
\begin{equation}\label{eqProdMoment1b}
\mbb{E}_{\pi^{(1)}}[\mb{q}_{m,n}]=\frac{1}{(R S)^{m+n}}\sum_{\tau\in\mc{S}_{m+n}}\mbox{card}\left\{(B,C)\mid C\circ(\tau^{-1}\eta_0)=C, B\circ\tau=B\right\}.
\end{equation}
The condition $C\circ(\tau^{-1}\eta_0)=C$ is equivalent to requiring $C$ to be constant on the orbits of $\tau^{-1}\eta_0$. For a given permutation $\tau^{-1}\eta_0$, there are $S^{\#(\tau^{-1}\eta_0)}$ ways to choose indexes $C$. Similarly, the condition $B=B\circ\tau$ is equivalent to requiring $B$ to be constant on the orbits of $\tau$ and there are $T^{\#(\tau)}$ ways to choose indexes B. As a result, (\ref{eqProdMoment1b}) equals
\begin{equation}\label{eqProdMoment1c}
\mbb{E}_{\pi^{(1)}}[\mb{q}_{m,n}]=\frac{1}{(R S)^{m+n}}\sum_{\tau\in\mc{S}_{m+n}}S^{\#(\tau^{-1}\eta_0)} T^{\#(\tau)}.
\end{equation}

Following the same procedures as in (\ref{eqPi1})-(\ref{eqProdMoment1c}), we obtain $\mbb{E}_{\pi_1^{(2)}}[\mb{q}_{m,n}]$ and $\mbb{E}_{\pi_2^{(2)}}[\mb{q}_{m,n}]$ as
\begin{align}
\mbb{E}_{\pi_1^{(2)}}[\mb{q}_{m,n}] &= \frac{1}{(R S)^m} \sum_{\tau_1\in\mc{S}_m} S^{\#(\tau_1^{-1}\eta_1)} T^{\#(\tau_1)},\label{eqProdMoment2}\\
\mbb{E}_{\pi_2^{(2)}}[\mb{q}_{m,n}] &= \frac{1}{(R S)^n} \sum_{\tau_2\in\mc{S}_n} S^{\#(\tau_2^{-1}\eta_2)} T^{\#(\tau_2)},\label{eqProdMoment3}
\end{align}
where the permutations $\eta_1=(1,\ldots,m)\in\mc{S}_m$ and $\eta_2=(1,\ldots,n)\in\mc{S}_n$. We multiply (\ref{eqProdMoment2}) with (\ref{eqProdMoment3}) and combine the permutations $\tau_1$ and $\tau_2$ to form a new permutation $\tau = \tau_1\circ\tau_2'\in\mc{S}_{m+n}$, where $\tau_2'$ is homogeneous to $\tau_2$ with $i$-th element relabeled as $m+i$. Note that $\pi_{\tau}\le\pi_{\eta_0}$, where partitions $\pi_{\tau}$ and $\pi_{\eta_0}$ are induced by $\tau$ and $\eta_0$, respectively. The new permutation $\tau$ is therefore $(m,n)$-disconnected, namely
\begin{equation*}\label{eqProdMoment4}
\mbb{E}_{\pi_1^{(2)}}[\mb{q}_{m,n}]\,\mbb{E}_{\pi_2^{(2)}}[\mb{q}_{m,n}] = \frac{1}{(R S)^{m+n}}\sum_{\substack{\tau\in\mc{S}_{m+n},\\ (m,n)\operatorname{-disconnected}}}S^{\#(\tau^{-1}\eta_0)} T^{\#(\tau)}.
\end{equation*}
Inserting $k_{m+n}(\mb{q}_{m,n})=\mbb{E}_{\pi^{(1)}}[\mb{q}_{m,n}]-\mbb{E}_{\pi^{(2)}}[\mb{q}_{m,n}]$ into (\ref{eqCumu2b}), we obtain
\begin{align*}
\kappa_{m,n} &= \lim_{R\rightarrow\infty}\frac{1}{S^{m+n}}\sum_{\substack{\tau\in\mc{S}_{m+n},\\ (m,n)\operatorname{-connected}}}S^{\#(\tau^{-1}\eta_0)+\#(\tau)} \left(\frac{T}{S}\right)^{\#(\tau)}.\label{eq2ndCumud}
\end{align*}
According to~(\ref{eqAnnGeo}), the exponent $\#(\tau)+\#(\tau^{-1}\eta_0)\le m+n$ for $\tau\in\mc{S}_{m+n}$ and is $(m,n)$-connected. In addition, the equality holds only when $\tau$ is non-crossing in the $(m,n)$-annular sense. Let $S\rightarrow\infty$, all terms in the summation with crossing permutation vanish and for $\tau\in\mc{S}_{\operatorname{a-nc}}(m,n)$, $S^{\#(\tau^{-1}\eta_0)+\#(\tau)}$ cancels with $S^{m+n}$. The derivation for the first order cumulants follows similarly.

\section{Proof of Proposition~\ref{propCLT}}\label{appxCLT}

Denote $\underline{\mc{G}_{\mb{Q}}}(z) = \mbb{E}_{\mb{\Psi}}\left[\widetilde{\mc{G}_{\mb{Q}}}(z)\right]$ as the expected resolvent of $\mb{Q}$, which is averaged over the ensembles of~$\mb{\Psi}$. As matrix $\mb{P}$ is random, $\underline{\mc{G}_{\mb{Q}}}$ is also a random variable and is the solution of (\ref{eqGQ}) with $F_{\mb{P}}(\cdot)$ replaced by its empirical version $\widetilde{F_{\mb{P}}}(\cdot)$, namely
\begin{equation}
z = \frac{1}{\underline{\mc{G}_{\mb{Q}}}} + \rho\int\frac{\lambda\mathrm{d}\widetilde{F_{\mb{P}}}(\lambda)}{1-\lambda \underline{\mc{G}_{\mb{Q}}}}.\label{eqGQT}
\end{equation}
We divide $G_R(z)$ into two parts as
\begin{equation}\label{eqGR}
G_R(z) = R\left(\widetilde{\mc{G}_\mb{Q}}(z)-\underline{\mc{G}_{\mb{Q}}}(z)\right) + R\left(\underline{\mc{G}_{\mb{Q}}}(z)-\mc{G}_{\mb{Q}}(z)\right) = G_R^1(z) + G_R^2(z).
\end{equation}
The proof of asymptotic Gaussianity of $G_R(z)$ follows in two steps, showing the asymptotic Gaussianity of $G_R^1(z)$ and $G_R^2(z)$, respectively. The proof then boils down to a direct application of Bai and Silverstein's lemma~\cite[Lemma 1.1]{Bai2004}:
\begin{lemma}(CLT of Wishart type ensembles~\cite{Bai2004}).\label{lemmaCovBai}
Consider an $N\times N$ Hermitian matrix $\mb{B} = \mb{X}^{\dag}\mb{T}\mb{X}/N$ and assume:
\begin{enumerate}
\item $\mb{X}$ is an $n\times N$ complex random matrix with \emph{i.i.d.} entries, $\mbb{E}[X_{i,j}] = 0$, $\mbb{E}[|X_{i,j}|^2] = 1$, and $\mbb{E}[|X_{i,j}|^4] = 2$;
\item $\mb{T}$ is a non-random Hermitian nonnegative definite matrix and its ESD $\widetilde{F_{\mb{T}}}(\cdot)$ converges weakly to a non-random limiting distribution $F_{\mb{T}}(\cdot)$.
\end{enumerate}
Let $M_n(z) = n\left(\widetilde{\mc{G}_{\mb{B}}}(z)-\mc{G}_{\mb{B}}(z)\right)$ and $\mc{C}_z$ a positive contour enclosing the support of $\mb{B}$, then the sequence $\left\{M_n(z)\right\}$ is tight on the contour $\mc{C}_z$, and $M_n(z)$ converges weakly to a Gaussian process on the complex plane with $\mbb{E}\left[M_n(z)\right] = 0$ and $\mathrm{Cov}(M_n(z_1),M_n(z_2))$ given by (\ref{eqGB}).
\end{lemma}

Conditioned on $\mb{P}$, it is straightforward to verify that the complex Gaussian matrix $\mb{\Psi}$ fulfills the assumption 1). The ESD $\widetilde{F_{\mb{P}}}(\cdot)$ converges to the Mar\v{c}enko-Pastur distribution and therefore fulfills the assumption 2). Furthermore, $\underline{\mc{G}_{\mb{Q}}}$ is, by definition, the average of $\widetilde{\mc{G}_{\mb{Q}}}$. It thus follows from Lemma~\ref{lemmaCovBai} that $G_R^1(z)$ given $\mb{P}$ converges to a Gaussian process on the complex plane with $\mbb{E}[G_R^1(z)]=0$ and
\begin{equation}
\mathrm{Cov}\left(G_R^1(z_1),G_R^1(z_2)\right) = \frac{\mc{G}_{\mb{Q}}'(z_1)\mc{G}_{\mb{Q}}'(z_2)}{(\mc{G}_{\mb{Q}}(z_1)-\mc{G}_{\mb{Q}}(z_2))^2} - \frac{1}{(z_1-z_2)^2}.\label{eqCovGR1}
\end{equation}

By (\ref{eqGQT}), we have
\begin{equation}
z = \frac{1}{\underline{\mc{G}_{\mb{Q}}}} + \rho\int\left(\frac{\lambda\mathrm{d}\widetilde{F_{\mb{P}}}(\lambda)}{1-\lambda \underline{\mc{G}_{\mb{Q}}}}-\frac{\lambda\mathrm{d}\widetilde{F_{\mb{P}}}(\lambda)}{1-\lambda \mc{G}_{\mb{Q}}}\right) + \rho\int\frac{\lambda\mathrm{d}\widetilde{F_{\mb{P}}}(\lambda)}{1-\lambda \mc{G}_{\mb{Q}}}.\label{eqGQTb}
\end{equation}
Subtracting (\ref{eqGQ}) from (\ref{eqGQTb}) yields
\begin{align}
0 &= \frac{\underline{\mc{G}_{\mb{Q}}}-\mc{G}_{\mb{Q}}}{\underline{\mc{G}_{\mb{Q}}}\mc{G}_{\mb{Q}}} - \rho\int\frac{\left(\underline{\mc{G}_{\mb{Q}}}-\mc{G}_{\mb{Q}}\right)\lambda^2\mathrm{d}\widetilde{F_{\mb{P}}}(\lambda)}{(1-\lambda\underline{\mc{G}_{\mb{Q}}})(1-\lambda\mc{G}_{\mb{Q}})}-\rho\int\left(\frac{\lambda\mathrm{d}\widetilde{F_{\mb{P}}}(\lambda)}{1-\lambda \mc{G}_{\mb{Q}}}-\frac{\lambda\mathrm{d}F_{\mb{P}}(\lambda)}{1-\lambda \mc{G}_{\mb{Q}}}\right)\nonumber\\
\underline{\mc{G}_{\mb{Q}}}-\mc{G}_{\mb{Q}} &= \frac{\rho\ \underline{\mc{G}_{\mb{Q}}}\mc{G}_{\mb{Q}}}{C}\int\left(\frac{\lambda\mathrm{d}\widetilde{F_{\mb{P}}}(\lambda)}{1-\lambda \mc{G}_{\mb{Q}}}-\frac{\lambda\mathrm{d}F_{\mb{P}}(\lambda)}{1-\lambda \mc{G}_{\mb{Q}}}\right),\label{eqGQdiff}
\end{align}
where $C = 1- \rho\ \underline{\mc{G}_{\mb{Q}}}\mc{G}_{\mb{Q}}\int\frac{\lambda^2\mathrm{d}\widetilde{F_{\mb{P}}}(\lambda)}{\left(1-\lambda\underline{\mc{G}_{\mb{Q}}}\right)\left(1-\lambda\mc{G}_{\mb{Q}}\right)}$.
By definition of Cauchy transform, we have
\begin{equation}
\int\frac{\lambda\mathrm{d}F_{\mb{P}}(\lambda)}{1-\lambda \mc{G}_{\mb{Q}}} = -\frac{1}{\mc{G}_{\mb{Q}}} + \frac{1}{\mc{G}_{\mb{Q}}^2}\mc{G}_{\mb{P}}\left(\frac{1}{\mc{G}_{\mb{Q}}}\right),\quad\mathrm{and}\quad \int\frac{\lambda\mathrm{d}\widetilde{F_{\mb{P}}}(\lambda)}{1-\lambda \mc{G}_{\mb{Q}}} = -\frac{1}{\mc{G}_{\mb{Q}}} + \frac{1}{\mc{G}_{\mb{Q}}^2}\widetilde{\mc{G}_{\mb{P}}}\left(\frac{1}{\mc{G}_{\mb{Q}}}\right).\label{eqGPQ}
\end{equation}
By inserting (\ref{eqGPQ}) into (\ref{eqGQdiff}) and multiplying $R$ on both sides of (\ref{eqGQdiff}), we obtain $G_R^2 = S\Big(\widetilde{\mc{G}_{\mb{P}}}(1/\mc{G}_{\mb{Q}})-\mc{G}_{\mb{P}}(1/\mc{G}_{\mb{Q}})\Big)/C$.
In the asymptotic regime (\ref{eqDim}), $\underline{\mc{G}_\mb{Q}}$ converges to $\mc{G}_{\mb{Q}}$ and $C$ converges to $1-\rho \mc{G}_{\mb{Q}}^2\int\frac{\lambda^2\mathrm{d}F_{\mb{P}}(\lambda)}{(1-\lambda\mc{G}_{\mb{Q}})^2}$. It follows from Lemma~\ref{lemmaCovBai} (with an $S\times S$ matrix $\mb{P} = \mb{X}^{\dag}\mb{X}$ and $\mb{T}$ being an identity matrix) that $G_R^2$ converges to a centered Gaussian process. Note that the covariance (\ref{eqCovGR1}) of $G_R^1(z)$ is independent of $\mb{P}$ and the randomness of $G_R^2(z)$ only comes from $\mb{P}$, which makes $G_R^1(z)$ and $G_R^2(z)$ independent of each other. Combining the above arguments, $G_R(z) = G_R^1(z) + G_R^2(z)$ is asymptotically a sum of two independent Gaussian processes and therefore $G_R(z)$ is also a Gaussian process.

\end{appendices}

\ifCLASSOPTIONcaptionsoff
  \newpage
\fi

\end{document}